%% file: main.tex
\documentclass[journal,comsoc, 10pt]{IEEEtran}
\usepackage[utf8]{inputenc}
\usepackage[T1]{fontenc}
\usepackage[
	disable,
backgroundcolor = White,textwidth=\marginparwidth]{todonotes}

\usepackage[markup=underlined]{changes}

\usepackage[utf8]{inputenc} 
\usepackage[T1]{fontenc}    
\usepackage{url}            
\usepackage{booktabs}       
\usepackage{amsfonts}       
\usepackage{nicefrac}       
\usepackage{microtype}      

\usepackage{pifont}
\usepackage{graphicx}
\usepackage{color}
\usepackage{rotating}
\usepackage{tabularx}
\usepackage{pdflscape}

\usepackage{amsmath,amssymb,amsthm}

\usepackage{blkarray}

\usepackage{algorithm, algorithmic}

\usepackage{bbm,dsfont}
\usepackage{caption,subcaption}
\usepackage{wrapfig}
\usepackage{cancel}
\usepackage{enumerate, cases}
\usepackage{thmtools,thm-restate}
\usepackage{mathtools}
\usepackage{url}

\usepackage[none]{hyphenat}
\usepackage{multirow}
\usepackage{makecell}
\usepackage{setspace}

\usepackage{dblfloatfix}
\usepackage{array}
\usepackage{enumitem}

\allowdisplaybreaks

\DeclareMathOperator*{\argmax}{arg\,max}

\newcolumntype{C}[1]{>{\centering}m{#1}}
\newcommand{\EE}[1]{\mathbb{E}\left[#1\right]}

\newcommand{\Prob}[1]{\mathbb{P}\left\{#1\right\}}

\newcommand{\floor}[1]{\left\lfloor #1 \right\rfloor}

\usepackage{graphicx}
\usepackage{mathptmx}
\usepackage{footnote}
\makesavenoteenv{tabular}
\usepackage{epstopdf}

\newtheorem{thm}{Theorem}
\newtheorem{lem}{Lemma}
\newtheorem{prop}{Proposition}
\newtheorem{cor}{Corollary}

\theoremstyle{definition}
\newtheorem{defn}{Definition}[section]

\newtheorem{rem}{Remark}

\usepackage{verbatim}

\usepackage{mathtools}

\DeclarePairedDelimiter\abs{\lvert}{\rvert}%
\makeatletter
\let\oldabs\abs
\def\abs{\@ifstar{\oldabs}{\oldabs*}}

\begin{document}
\title{
Learning Optimal Phase-Shifts of Holographic Metasurface Transceivers}

\author{Debamita Ghosh, IITB-Monash Research Academy, IIT Bombay, India \\
Manjesh K. Hanawal, MLioNS Lab, IEOR, IIT Bombay, India \\
Nikola Zlatanov, Innopolis University, Russia}
\maketitle

\begin{abstract}

Holographic metasurface transceivers (HMT) is an emerging technology for enhancing the coverage and rate of wireless communication systems. However, acquiring accurate channel state information in HMT-assisted wireless communication systems is critical for achieving these goals. In this paper, we propose an algorithm for learning the optimal phase-shifts at a HMT for the far-field channel model. Our proposed algorithm exploits the structure of the channel gains in the far-field regions and learns the optimal phase-shifts in presence of noise in the received signals. We prove that the probability that the optimal phase-shifts estimated by our proposed algorithm deviate from the true values decays exponentially in the number of pilot signals. Extensive numerical simulations validate the theoretical guarantees and also demonstrate significant gains as compared to the state-of-the-art policies. 
\end{abstract}

\begin{IEEEkeywords}
 Holographic Metasurface Transceivers, Channel State Information, Uniform Exploration
		 \vspace{-3mm}
\end{IEEEkeywords}

	
\section{Introduction}
Future wireless network technologies, namely beyond-5G and 6G, have been focused on millimeter wave (mmWave) and TeraHertz (THz) communications technologies as possible solutions to the ever growing demands for higher data rates and lower latency. However, mmWave and THz communications have challenges that need to be addressed before this technology is adopted \cite{wan2021terahertz}, \cite{jamali2020intelligent}. One such major challenge is signal deterioration due to reflections and absorption. 

A possible solution for the signal deterioration are base stations (BSs) with massive antennas arrays that can provide large beamforming gains and thereby compensate for the signal deterioration \cite{larsson2014massive}.  However, implementing a BS with a massive antenna array is itself challenging due to the high hardware costs. Holographic Metasurface Transceivers (HMTs) are introduced as a promising solution for building a massive antenna array \cite{wu2019intelligent, huang2020holographic}. A HMT is comprised of a large number of metamaterial elements densely deployed into a limited surface area in order to form a spatially continuous transceiver aperture. These metamaterial elements at the HMT acts as phase-shifting antennas, where each phase-shifting element of the HMT can change the phase of transmiting/receiving signal 
and thereby  beamform 
towards desired directions where the users are allocated \cite{hu2018beyond}. Due to these continuous apertures, HMTs can be represented as an extension of the traditional massive antenna arrays with discrete antennas to continuous reflecting surfaces \cite{hu2018beyond}. 

In this paper, we consider the HMT-assisted wireless systems illustrated in Fig. \ref{fig:schematic_figure}, where a HMT acts as a BS that serves multiple users. The performance of this system is dependent on channel state information (CSI) estimates at the HMT, which are used for accurate beamforming towards the users. The authors in \cite{yoo2021holographic} and \cite{zhang2022beam} have studied the effect of HMT-assisted systems on enhancing the communication performance under the assumption of perfect CSI. However, perfect CSI is not available in practice. In practice, the CSI has to be estimated via pilot signals, which results in inaccurate CSI estimates at the HMT.

The aim of this paper is to obtain accurate CSI estimates at the HMT, which in turn is used to set the optimal phase-shifts at the HMT that maximize the data rate to the users when the users are located in the far-field. 
To this end, we exploit the structure of the far-field channel model between the HMT and the users to show that the optimal phase-shifts at the HMT can be obtained from five samples of the received pilot signals at the HMT in a noiseless environment. We then use this approach to develop a learning algorithm that learns the optimal phase-shifts from the received pilot signals at the HMT in a noisy environment. 
Finally, we provide theoretical guarantees for our learning algorithm. Specifically, we prove that the probability of the phase-shifts generated by our algorithm to deviate by more than $\epsilon$ from the optimal phase-shifts is small and decays as the number of pilot symbols increases. The error analysis is based on tail probabilities of the non-central Chi-squared distribution. 

In summary, our main contributions are as follows:
\begin{itemize}
\item We propose an efficient learning algorithm for estimating the optimal phase-shifts at an HMT in the presence of noise for the case when the users that the HMT is serving are located at the far-field region.
\item We prove that the probability of the phase-shifts generated by our algorithm to deviate by more than $\epsilon$ from the optimal phase-shifts is small and decays exponentially as the number of pilots used for estimation increases. 
\item We show numerically that the performance of the proposed algorithm significantly outperforms existing CSI estimation algorithms.
\end{itemize}

\subsection{Related Works}
Several channel estimation schemes, which are proposed for the massive antenna arrays, are also applicable to the considered HMT including exhaustive search \cite{dai2006efficient}, hierarchical search \cite{xiao2016hierarchical}, \cite{chen2018beam}, and compressed sensing (CS) \cite{chen2018beam}. As the exhaustive search in \cite{dai2006efficient} significantly increases the training overhead, the authors in \cite{xiao2016hierarchical} and \cite{chen2018beam} proposed the hierarchical search based on a predefined codebook as an improvement over the exhaustive search. The hierarchical schemes, in general, may incur high training overhead and system latency since they require non-trivial coordination among the transmitter and the user \cite{chen2018beam}. On the other hand, the proposed CS-based channel estimation scheme in \cite{chen2018beam} provides trade-offs between accuracy of estimation and training overhead at different computational costs. 

On the other hand, CSI estimation  schemes developed specifically for HMTs can be found in \cite{WCL2022channelestimationHMMIMO} and \cite{ghermezcheshmeh2021channel}. 
The authors in \cite{WCL2022channelestimationHMMIMO} proposed the least-square estimation based approach to study the channel estimation problem for the uplink between a single user and the BS equipped with the holographic surface with a large number of antennas. However, the authors require an additional knowledge of antennas array geometry to reduce the pilot overhead required by the channel estimation, and hence the computational complexity scales up with the number of antennas at the BS.
In \cite{ghermezcheshmeh2021channel}, the authors proposed a scheme for the estimation of the far-field channel between a HMT and a user that requires only five pilots for perfect estimation 
in the noise-free environment. In the noisy case, the authors of \cite{ghermezcheshmeh2021channel} proposed an iterative algorithm that efficiently estimates the far-field channel. 
Unlike the existing works, the training overhead and the computational cost of the proposed scheme in \cite{ghermezcheshmeh2021channel} does not scale with the number of phase-shifting elements at the HMT. The iterative algorithm in \cite{ghermezcheshmeh2021channel} significantly outperforms the hierarchical and CS based schemes. However, the authors in \cite{ghermezcheshmeh2021channel} did not provide any theoretical guarantees on their proposed algorithm. Motivated by \cite{ghermezcheshmeh2021channel}, in this work, we propose an algorithm which outperforms the one in \cite{ghermezcheshmeh2021channel}, and, in addition, we also provide theoretical guarantees for our proposed algorithm. 
 
This paper is organized as follows. The system and channel models for the HMT communication system are given in Sec. \ref{sec:Model}. The proposed algorithm for learning the optimal phase-shifts is given in Sec. \ref{sec:EstimationStrategy} and its theoretical guarantee is provided in Sec. \ref{sec:Analysis}. Numerical evaluation of the proposed algorithm is provided in Sec. \ref{sec:Numericals}. Finally, Sec. \ref{sec:Conclusions} concludes the paper.

\section{System and Channel Models}
\label{sec:Model}
\input{Model}

\section{Proposed Channel Estimation Strategy}
\label{sec:EstimationStrategy}
\input{EstimationStrategy}

\section{Theoretical Guarantees For The Proposed Algorithm}
\label{sec:Analysis}
\input{Analysis}

\section{Numerical Simulations}
\label{sec:Numericals}
\input{Numericals}

\section{Conclusion}
\label{sec:Conclusions}
\input{Conclusions}

\section*{Appendix}
\label{sec:Appendix}
\input{Appendix}

\bibliographystyle{IEEEtran}
\bibliography{ref}

\end{document}

%% file: Model.tex
We consider a HMT-assisted wireless communication system, shown in Fig. \ref{fig:schematic_figure}, where an HMT communicates with multiple users in the mmWave band. We assume that there is a Line of Sight (LoS) between the HMT and each user. As a result, when modeling the far-field channel, we only take into account the LoS path since its power is order of magnitude higher than non-line-of-sight (NLoS) paths \cite{akdeniz2014millimeter}. The NLoS components are incorporated in the noise. We assume that the users send orthogonal pilots to the HMT for channel estimation. Based on the estimated CSI at the HMT to each user, the HMT sends data to the users. Hence, the data rate from the HMT to the users is directly dependent on the accuracy of the CSI estimates at the HMT. Since in this paper our main goal is the accurate CSI estimation at the HMT to each user, which in turn send orthogonal pilots to the HMT, in the rest of the paper, we will focus on the CSI estimation between the HMT and a typical user.
\begin{figure}[ht]
\vspace{-3mm}
    \centering
    \includegraphics[scale = 0.63, trim = 5cm 15.5cm 4cm 2cm, clip]{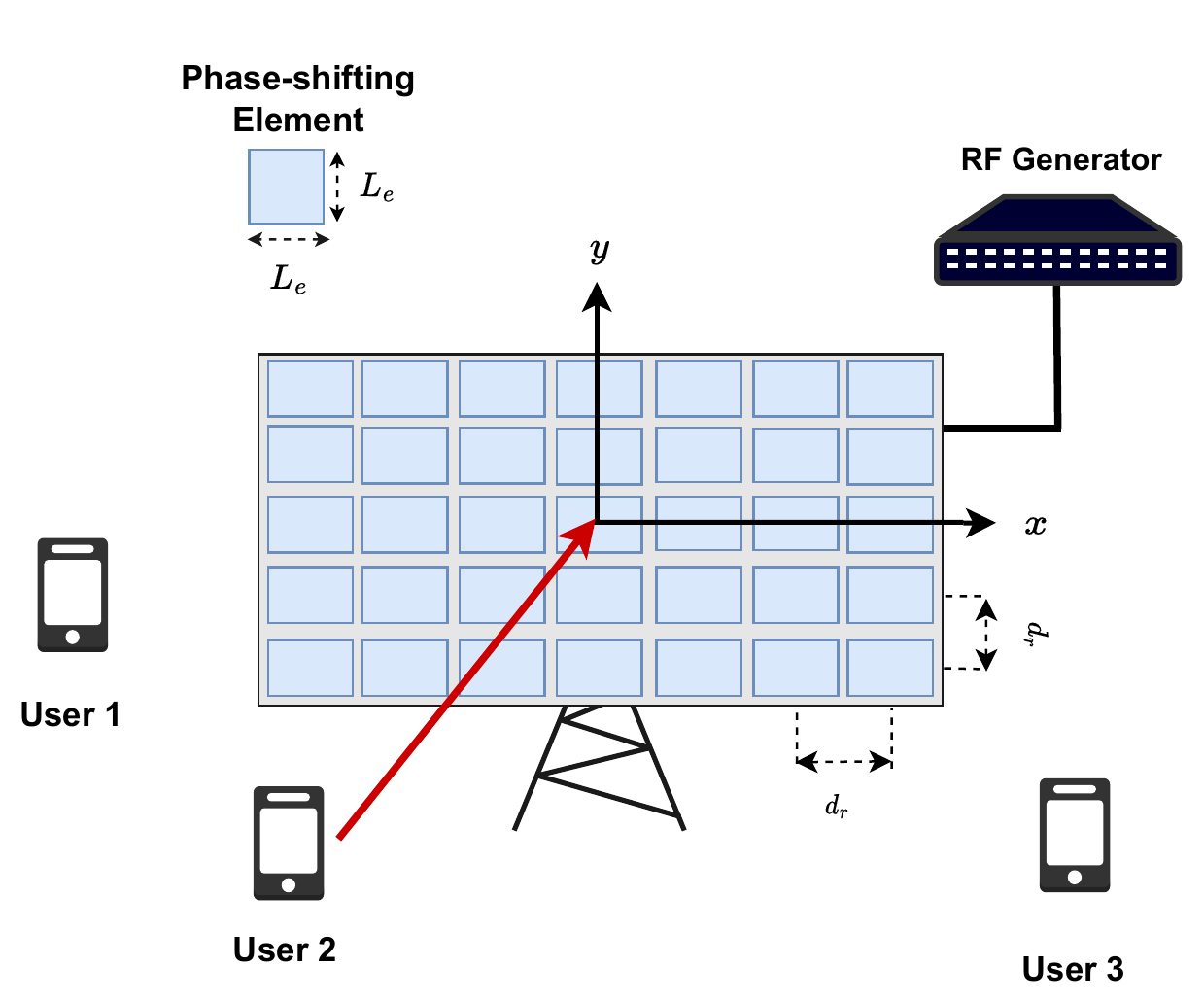} 
    \caption{\small{The HMT-assisted wireless communication system \cite{ghermezcheshmeh2021channel}.}}
    \label{fig:schematic_figure}
   \vspace{-3mm}
\end{figure}

\subsection{HMT Model}

The HMT has a rectangular surface of size $L_x \times L_y$, where $L_x$ and $L_y$ are the width and the length of the surface, respectively. The HMT's surface is comprised of a large number of sub-wavelength phase-shifting elements, where each elements is assumed to be a square of size $L_e \times L_e$ and can change the phase of the transmit/receive signal independently from rest of the elements. Let $d_r$ be the distance between two neighboring phase-shifting elements. The total number of phase-shifting elements of the HMT is given by $M = M_x \times M_y$, where $M_x = L_x/d_r$ and $M_y = L_y/d_r$. Without loss of generality, we assume that the HMT lies in the $x - y$ plane of a Cartesian coordinate system, where the center of the surface is at the origin of the coordinate system. Assuming $M_x$ and $M_y$ are odd numbers, the position of the $(m_x, m_y)^{th}$ phase-shifting element in the Cartesian coordinate system is given as $(x, y) = (m_xd_r, m_yd_r),$ where $m_x \in \left \{ -\frac{M_x-1}{2},\dots,\frac{M_x-1}{2} \right \}$ and $m_y \in \left \{ -\frac{M_y-1}{2},\dots,\frac{M_y-1}{2} \right \}$. When $M_x$ or $M_y$ is even, the position of the $(m_x,m_y)^{th}$ element can be appropriately defined. 
\vspace{-3mm}

\subsection{Channel Model}
 
Consider the channel between the $(m_x,m_y)^{th}$ phase-shifting element at the HMT and the typical user. Let the beamforming weight imposed by the $(m_x,m_y)^{th}$ phase-shifting element at the HMT be $\Gamma_{m_xm_y} = e^{j\beta_{m_xm_y}}$, where $\beta_{m_xm_y}$ is the phase shift at the $(m_x, m_y)^{th}$ element. Let $\lambda$ denote the wavelength of the carrier frequency, $k_0 = \frac{2\pi}{\lambda}$ be the wave number, $d_0$ be the distance between the user and the center of the HMT and let $F_{m_xm_y}$ denote the effect of the size and power radiation pattern of the $(m_x,m_y)^{th}$ phase-shifting element on the channel coefficient \cite{ellingson2021path}. Due to the far-field assumptions, the radiation pattern of all the phase-shifting elements of the HMT are identical, i.e.,  $F_{m_x m_y} = F,\; \forall m_x, m_y$ holds. Finally, let $\theta$ and $\phi$ denote the elevation and azimuth angles of the impinging wave from the user to the center of the HMT, see Fig. \ref{fig:distance_figure}.
\begin{figure}[ht]
\vspace{-5mm}
    \centering
    \includegraphics[scale = 0.75, trim = {6cm 18.5cm 3cm 4cm}, clip]{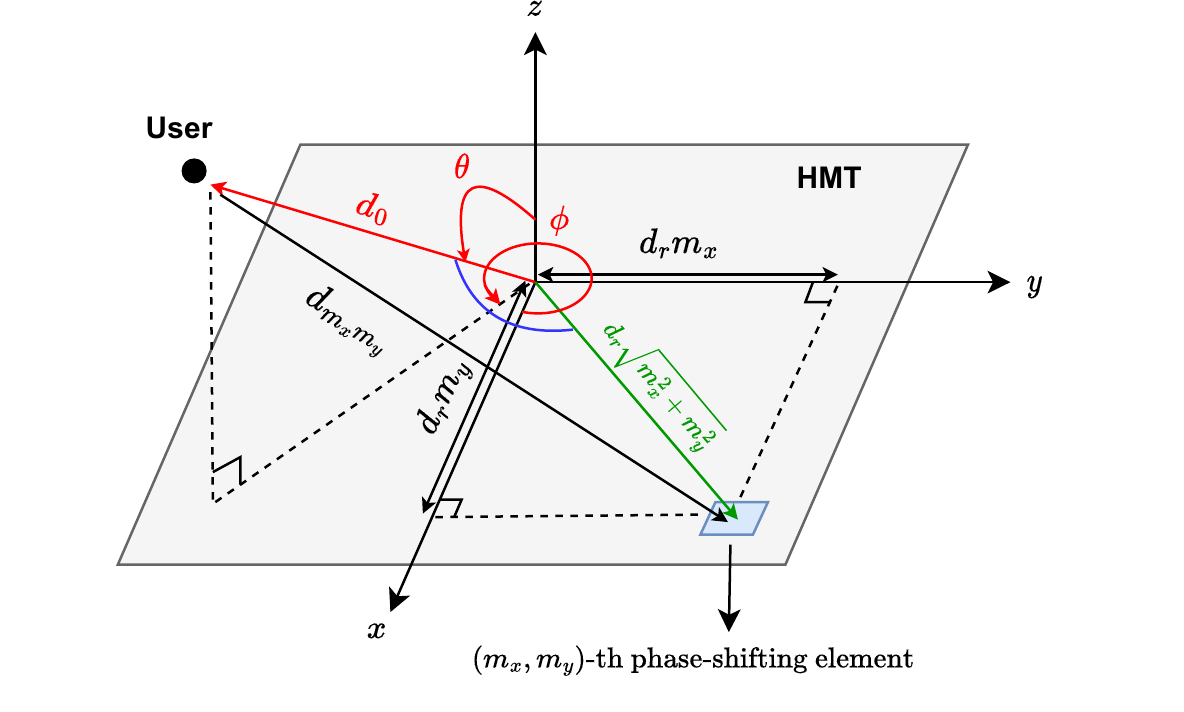}
    \caption{\small{Distance between the $(m_x,m_y)$-th phase-shifting element at the HMT and the user \cite{ghermezcheshmeh2021channel}.}}
    \label{fig:distance_figure}
    \vspace{-3mm}
\end{figure}

Now, if the phase-shift imposed by the $(m_x,m_y)^{th}$ element, $\beta_{m_x,m_y}$, is set to
\begin{equation}
\beta_{m_xm_y} = -\mod(k_0d_r(m_x\beta_1 + m_y\beta_2), 2\pi), \forall m_x, m_y,\nonumber
\end{equation}
where $\beta_1$ and $\beta_2$ are the phase-shift parameters \cite{ghermezcheshmeh2021channel}, \cite{selvan2017fraunhofer}, \cite{najafi2020physics}, which are the only degrees of freedom within the phase-shift $\beta_{m_xm_y}$, then 
the HMT-user channel in the far-field is approximated accurately by \cite{ghermezcheshmeh2021channel}, \cite{selvan2017fraunhofer}, \cite{najafi2020physics} 
\begin{align}
    H(\beta_1, \beta_2) &= \left(\frac{\sqrt{F}\lambda e^{-jk_0d_0}}{4\pi d_0}\right)L_xL_y\times \mathrm{sinc}\bigg(K_x\pi(\alpha_1 - \beta_1)\bigg)\nonumber\\
    &\qquad \qquad \times\mathrm{sinc}\bigg(K_y\pi(\alpha_2 - \beta_2)\bigg),\label{eqn:contphaseHMMIMOuserchannel}
\end{align}
where $K_x = \frac{L_x}{\lambda}, K_y = \frac{L_y}{\lambda}, \alpha_1 = \sin(\theta)\cos(\phi), \alpha_2 = \sin(\theta)\sin(\phi),$ and  $\mathrm{sinc}(x) = \frac{\sin(x)}{x}$. Please note that $\alpha_1 \in [-1,1]$ and $\alpha_2 \in [-1,1]$, and their values depend on the location of the user, i.e., on $\theta$ and $\phi$.

From \eqref{eqn:contphaseHMMIMOuserchannel}, it is clear that the absolute value of the HMT-user channel is maximized when the two sinc functions attain their maximum values, which occurs when the phase-shifting parameters, $\beta_1$ and $\beta_2$, are set to $\beta_1 = \alpha_1$ and $\beta_2 = \alpha_2$, where $(\alpha_1,\alpha_2)$ are unknown to the HMT since they depend on the location of the user. Therefore, in the far-field case, the problem of finding the optimal phase-shifts of the elements at the HMT reduces to estimating the two parameters, $\alpha_1$ and $\alpha_2$ at the HMT. 

\begin{rem}
Fig. \ref{fig:abs_h} shows an example of $|H(\beta_1,\beta_2)|$ as a function of $(\beta_1, \beta_2)$. As can be seen from Fig. \ref{fig:abs_h}, the graph of $|H(\beta_1,\beta_2)|$ hits zero periodically and has several lobes. The optimal value $(\alpha_1,\alpha_2) = (0.68,-0.45)$ is attained at the central lobe which has the highest peak and is attained for $(\beta^*_1,\beta^*_2) = (\alpha_1,\alpha_2) = (0.68, -0.45).$
\end{rem}
\begin{figure}[ht]
    \centering
    \includegraphics[width = 8.5cm, height = 4.5cm]{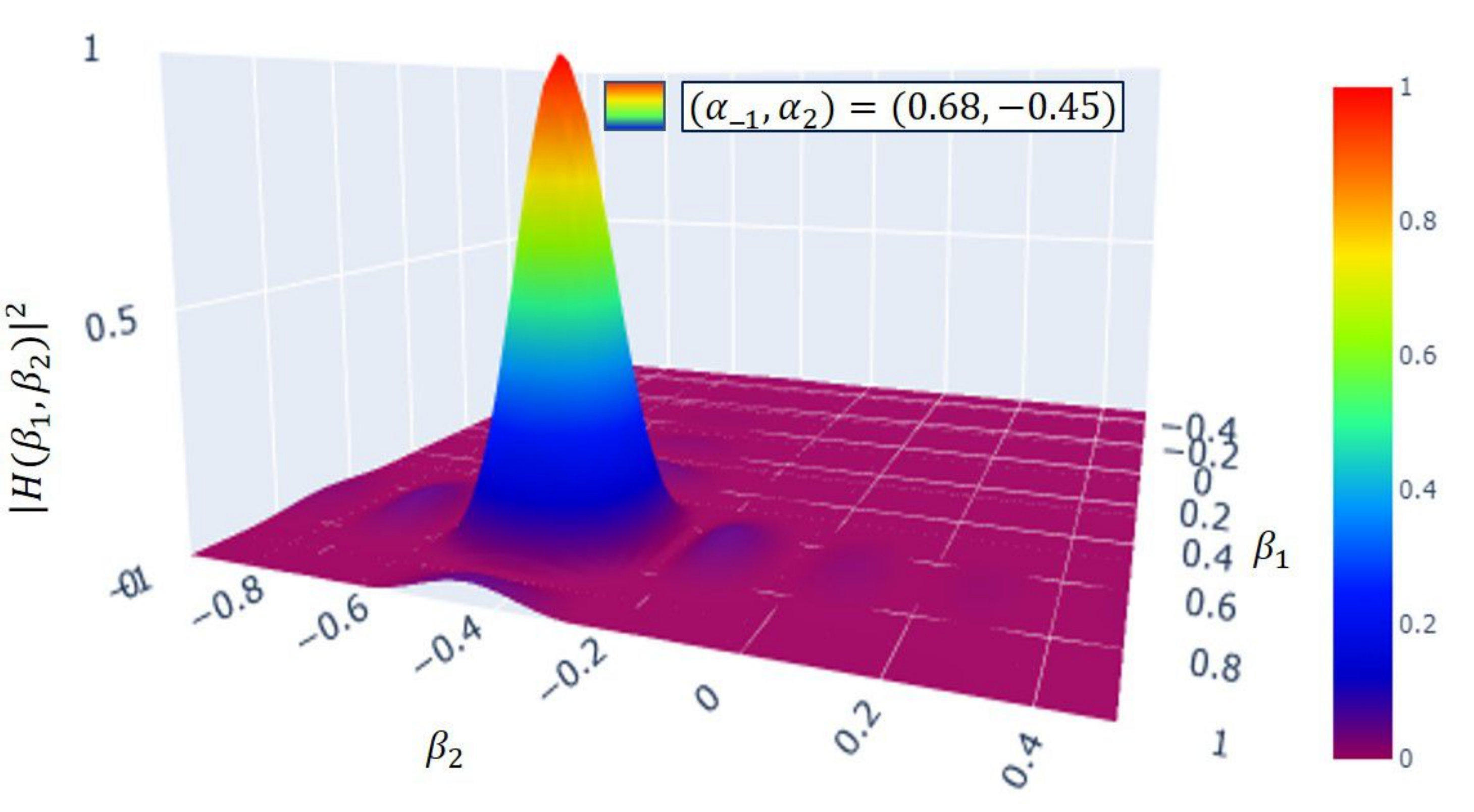}
    \caption{\small{$|H(\beta_1,\beta_2)|$ v/s $(\beta_1,\beta_2)$ for values of $(\alpha_1,\alpha_2) = (0.68,-0.45).$}}
    \label{fig:abs_h}
\vspace{-3mm}
\end{figure}
\vspace{-3mm}

%% file: EstimationStrategy.tex
In this section, we propose an algorithm that estimates the optimal phase-shifting parameters $\beta_1$ and $\beta_2$ that maximize $\abs{H(\beta_1,\beta_2)}$ in \eqref{eqn:contphaseHMMIMOuserchannel} in the presence of noise. 
\subsection{Problem Formulation}
In the channel estimation procedure, the user sends a pilot symbol $x_p = \sqrt{P}$ to the HMT, where $P$ is the pilot transmit power. Then, the received signal at the HMT for fixed phase-shifting parameters $(\beta_1,\beta_2),$ denoted by $y(\beta_1,\beta_2)$, is given by 
\begin{align}
   y(\beta_1,\beta_2) & = \sqrt{P}\times H(\beta_1,\beta_2) + \zeta,  \label{eqn:receivedsignal}
 \end{align}
 where $\zeta$  is the complex-valued additive white Gaussian noise (AWGN) with zero mean and variance $\sigma^2$ at the HMT. The received signal in \eqref{eqn:receivedsignal} is then squared in order to obtain the received signal squared, denoted by $r(\beta_1,\beta_2)$, and given by
 \begin{align}
    r(\beta_1,\beta_2) &= \abs{y(\beta_1,\beta_2)}^2 = \abs{ \sqrt{P}\times H(\beta_1,\beta_2) + \zeta}^2. \label{eqn:absRSS}
\end{align}

\noindent
\textbf{Objective:} Our goal is to identify the optimal phase-shifting parameters, denoted by $(\beta_1^*,\beta_2^*)$, at the HMT that maximizes $r(\beta_1,\beta_2)$ given by \eqref{eqn:absRSS}. Specifically, we aim to solve the following optimisation problem
\begin{align}
\label{eqn:objfn}
    (\beta_1^*,\beta_2^*) &= \argmax_{\substack{\beta_1 \in [-1,1]\\ \beta_2 \in [-1,1]}}r(\beta_1,\beta_2).
\end{align} 
The expected value of $r(\beta_1,\beta_2)$, denoted by $\mu(\beta_1,\beta_2)$, is given by
\begin{align}
    \mu(\beta_1,\beta_2) &= \EE{r(\beta_1,\beta_2)} \nonumber\\
    &= \abs{\sqrt{P}\times H(\beta_1,\beta_2)}^2 + \sigma^2.\label{eqn:expectedRSS}
\end{align}
Using \eqref{eqn:expectedRSS}, the optimization problem in \eqref{eqn:objfn} can be written equivalently as
\begin{align}
\label{eqn:expobjfn}
    (\beta_1^*,\beta_2^*) &= \argmax_{\substack{\beta_1 \in [-1,1]\\ \beta_2 \in [-1,1]}}\mu(\beta_1,\beta_2).
\end{align}
In order to obtain an intuition on how to solve \eqref{eqn:expobjfn}, we first assume that $\mu(\beta_1,\beta_2)$ in \eqref{eqn:expectedRSS} is known perfectly at the HMT for five specific values of the pair $(\beta_1,\beta_2)$. Later, we use the same intuition to solve \eqref{eqn:expobjfn} when $\mu(\beta_1,\beta_2)$ are not known perfectly but can be estimated. 

\subsection{The Optimal Phase-Shifting Parameters When $\mu(\beta_1,\beta_2)$ Are Known In Advance} 
For notational convenience, let us define the set $\mathcal{B}$ as 
\begin{align}
\label{eqn:setpar}
\mathcal{B} = \bigg\{&(\beta^0_1, \beta^0_2), (\beta^0_1 + v, \beta^0_2), (\beta^0_1 - v, \beta^0_2), \nonumber\\
& (\beta^0_1, \beta^0_2 + w), (\beta^0_1, \beta^0_2 - w)\bigg\}.
\end{align}
The set $\mathcal{B}$ is comprised of five pairs of the phase-shifting parameters $(\beta_1,\beta_2),$ where $\beta^0_1$ and $\beta^0_2$ are some initial arbitrarily selected phase-shifting parameters, $v$ and $w$ are numbers chosen such that $K_xv \in \mathbb{N}$ and $K_yw \in \mathbb{N}$ hold, where $\mathbb{N}$ is the set of natural numbers. Please note that for a selected $(\beta^0_1,\beta^0_2)$ and a chosen $v$ and $w$, if $\abs{\beta^0_1 \pm v} \geq 1$ then we set $\abs{\beta^0_1 \pm v} = 1.$ In the same way, if $\abs{\beta^0_2 \pm w} \geq 1$ then we set $\abs{\beta^0_2 \pm w} = 1.$

\begin{thm}
\label{thm:actualmeansestimate}
If the HMT can obtain $\mu(\beta^0_1, \beta^0_2)$, $\mu(\beta^0_1 + v, \beta^0_2)$, $\mu(\beta^0_1 - v, \beta^0_2)$, $\mu(\beta^0_1, \beta^0_2 + w)$ and $\mu(\beta^0_1, \beta^0_2 - w)$, i.e., obtain $\mu(\beta_1,\beta_2)$ for the five phase-shifting parameters in $(\beta_1,\beta_2) \in \mathcal{B}$ given in \eqref{eqn:setpar}, then the optimal phase-shifting parameters $\beta_1^*$ and $\beta_2^*,$ which are the solutions of \eqref{eqn:expobjfn}, are given by
 \begin{align}
        &\beta_1^* = \left\{\frac{\alpha_1^{(i)}  + \alpha_1^{(j)}}{2} : \min\limits_{i \in \{1,2\}, j \in \{3,4\}}\abs{\alpha_1^{(i)} - \alpha_1^{(j)}}\right\}\label{eqn:beta1opt},
        \intertext{where }
            & \alpha_1^{(1)/(2)} = \beta^0_1 + \frac{v}{ 1 \pm \sqrt{\abs{\frac{\mu\left(\beta^0_1, \beta^0_2\right) - \sigma^2}{\mu\left(\beta^0_1+v, \beta^0_2\right) - \sigma^2}}}}\nonumber\\ 
            & \alpha_1^{(3)/(4)} = \beta^0_1 - \frac{v}{ 1 \pm \sqrt{\abs{\frac{\mu\left(\beta^0_1, \beta^0_2\right) - \sigma^2}{\mu\left(\beta^0_1-v, \beta^0_2\right) - \sigma^2}}}}\nonumber\\
            \intertext{and}
         &\beta_2^* = \left\{\frac{\alpha_2^{(i)}  + \alpha_2^{(j)}}{2} : \min\limits_{i \in \{1,2\}, j \in \{3,4\}}\abs{\alpha_2^{(i)} - \alpha_2^{(j)}}\right\}\label{eqn:beta2opt},\\
         \intertext{where }
            & \alpha_2^{(1)/(2)} = \beta^0_2 + \frac{v}{ 1 \pm \sqrt{\abs{\frac{\mu\left(\beta^0_1, \beta^0_2\right) - \sigma^2}{\mu\left(\beta^0_1, \beta^0_2+w\right) - \sigma^2}}}}\nonumber\\
            & \alpha_2^{(3)/(4)} = \beta^0_2 - \frac{v}{ 1 \pm \sqrt{\abs{\frac{\mu\left(\beta^0_1, \beta^0_2\right) - \sigma^2}{\mu\left(\beta^0_1, \beta^0_2-w\right) - \sigma^2}}}}\nonumber
        \end{align}
\end{thm}

\begin{proof}
By using \eqref{eqn:expectedRSS} and \eqref{eqn:contphaseHMMIMOuserchannel} for any $(\beta_1,\beta_2) = (\beta^0_1, \beta^0_2)$, we have the following
\begin{align}
    \mu(\beta^0_1, \beta^0_2) - \sigma^2 = \Bigg|&\sqrt{P}\left(\frac{\sqrt{F}\lambda e^{-jk_0d_0}}{4\pi d_0}\right)L_xL_y\mathrm{sinc}\bigg(K_x(\alpha_1 - \beta^0_1)\bigg)\nonumber\\
    &\times\mathrm{sinc}\bigg(K_y(\alpha_2 - \beta^0_2)\bigg)\Bigg|^2\label{eqn:noiseeqn1}.
\end{align}
For $(\beta_1,\beta_2) = (\beta^0_1 + v, \beta^0_2)$, where $v$ is any arbitrary parameter such that $K_xv \in \mathbb{N}$ and $\abs{\beta^0_1 \pm v} \leq 1$ holds, we have
\begin{align}
    \mu(\beta^0_1 + v, \beta^0_2) - \sigma^2 = \Bigg|&\sqrt{P}\bigg(\frac{\sqrt{F}\lambda e^{-jk_0d_0}}{4\pi d_0}\bigg)L_xL_y\nonumber\\
    &\times\mathrm{sinc}\bigg(K_y(\alpha_2 - \beta^0_2)\bigg)\Bigg|^2.\label{eqn:noiseeqn2}
\end{align}
Dividing \eqref{eqn:noiseeqn1} by \eqref{eqn:noiseeqn2}, we obtain
\begin{align}
    &\frac{\mu(\beta^0_1, \beta^0_2) - \sigma^2}{ \mu(\beta^0_1 + v, \beta^0_2) - \sigma^2} = \frac{\abs{\mathrm{sinc}\bigg(K_x(\alpha_1 - \beta^0_1)\bigg)}^2}{\abs{\mathrm{sinc}\bigg(K_x(\alpha_1 - \beta^0_1 - v)\bigg)}^2} \nonumber\\
   &\frac{\mu(\beta^0_1, \beta^0_2) - \sigma^2}{ \mu(\beta^0_1 + v, \beta^0_2) - \sigma^2} = \frac{\abs{\frac{\sin\left(K_x \pi(\alpha_1 - \beta^0_1)\right)}{K_x\pi(\alpha_1 - \beta^0_1)}}^2}{\abs{\frac{\sin\left(K_x \pi(\alpha_1 - \beta^0_1-v)\right)}{K_x\pi(\alpha_1 - \beta^0_1-v)}}^2}.\label{eqn:rationoise1}
\end{align}
If $v$ is selected such that $K_xv \in \mathbb{N}$, then we have $\abs{\sin\bigg(K_x\pi(\alpha_1 - \beta^0_1 \pm v)\bigg)} = \abs{\sin\bigg(K_x\pi(\alpha_1 - \beta^0_1)\bigg)}.$ As a result, \eqref{eqn:rationoise1} is simplified to
\begin{align}
    \frac{\mu(\beta^0_1, \beta^0_2) - \sigma^2}{ \mu(\beta^0_1 + v, \beta^0_2) - \sigma^2} = \abs{\frac{\alpha_1 - \beta^0_1 - v}{\alpha_1 - \beta^0_1}}^2.\label{eqn:intstep}
\end{align}
Since $\mu(\beta_1,\beta_2) \geq \sigma^2$,  it follows that $\mu(\beta_1,\beta_2)-\sigma^2 = \abs{\mu(\beta_1,\beta_2)-\sigma^2}$ always holds, for all $(\beta_1,\beta_2) \in \mathcal{B}$. Using this fact, \eqref{eqn:intstep} can be written equivalently as
\begin{align}
    \sqrt{\abs{\frac{\mu(\beta^0_1, \beta^0_2) - \sigma^2}{ \mu(\beta^0_1 + v, \beta^0_2) - \sigma^2}}} = \abs{\frac{\alpha_1 - \beta^0_1 - v}{\alpha_1 - \beta^0_1}}. \label{eqn:nonlinear1}
\end{align}
By solving the nonlinear equation in  \eqref{eqn:nonlinear1} w.r.t. the unknown $\alpha_1$, we obtain two solutions for $\alpha_1$, denoted by $\alpha_1^{(1)}$ and $\alpha_1^{(2)}$, given by
\begin{align}
    &\alpha_1^{(1)/(2)} = \beta^0_1 +  \frac{v}{1 \pm \sqrt{\abs{\frac{\mu(\beta^0_1, \beta^0_2) - \sigma^2}{\mu(\beta^0_1 + v, \beta^0_2) - \sigma^2}}}}.  \label{eqn:alpha1solutionratio1}
\end{align}
It is not known which of the two values $\alpha_1^{(1)}$ and $\alpha_1^{(2)}$ is equal to $\alpha_1.$ To identify the correct solution for $\alpha_1$ of the two solutions given by \eqref{eqn:alpha1solutionratio1}, we need the value of $\mu(\beta_1, \beta_2)$ for $(\beta_1,\beta_2) = (\beta^0_1-v,\beta^0_2).$ Following the same procedure as for \eqref{eqn:noiseeqn1}-\eqref{eqn:alpha1solutionratio1}, but now by using the values of $\mu(\beta_1, \beta_2)$ for $(\beta_1,\beta_2) = (\beta^0_1, \beta^0_2)$ and $(\beta_1,\beta_2) = (\beta^0_1 - v, \beta^0_2)$, we obtain
\begin{align}
    \sqrt{\abs{\frac{\mu(\beta^0_1, \beta^0_2) - \sigma^2}{ \mu(\beta^0_1 - v, \beta^0_2) - \sigma^2}}} = \abs{\frac{\alpha_1 - \beta^0_1 + v}{\alpha_1 - \beta^0_1}}. \label{eqn:nonlinear2}
\end{align}
By solving \eqref{eqn:nonlinear2}, we obtain
\begin{align}
    \alpha_1^{(3)/(4)} = \beta^0_1 -  \frac{v}{1 \pm \sqrt{\abs{\frac{\mu(\beta^0_1, \beta^0_2) - \sigma^2}{\mu(\beta^0_1 - v, \beta^0_2) - \sigma^2}}}}. \label{eqn:alpha1solutionratio2}
\end{align}
One of the solutions in \eqref{eqn:alpha1solutionratio1} is identical to one of the solutions in \eqref{eqn:alpha1solutionratio2}. Therefore, using \eqref{eqn:alpha1solutionratio1} and \eqref{eqn:alpha1solutionratio2}, the correct solution of $\alpha_1$ can be obtained as\footnote{Note that $\alpha_1$ can also be written equivalently as $\alpha_1 = \left\{\alpha_1^{(1)}, \alpha_1^{(2)}\right\} \bigcap \left\{\alpha_1^{(3)}, \alpha_1^{(4)}\right\}$. However, the expression in \eqref{eqn:noiseestalpha1final} is more convenient for the case when the values of $\mu(\beta_1,\beta_2)$ need to be estimated.}
\begin{align}
\alpha_1 = \left\{\frac{\alpha_1^{(i)}  + \alpha_1^{(j)}}{2} : \min_{i \in \{1,2\}, j \in \{3,4\}}\abs{\alpha_1^{(i)} - \alpha_1^{(j)}}\right\}.\label{eqn:noiseestalpha1final} 
\end{align}
In order to obtain $\alpha_2$, we need the value of $\mu(\beta_1,\beta_2)$ for $(\beta_1, \beta_2) = (\beta^0_1, \beta^0_2)$, which we already have, and for $(\beta_1, \beta_2) = (\beta^0_1, \beta^0_2 + w)$, where $w$ is selected such that $K_yw \in \mathbb{N}$, $\abs{\beta^0_2 \pm w} \leq 1$ and $\abs{\sin(K_y\pi(\alpha_2 - \beta^0_2 \pm w))} = \abs{\sin(K_y\pi(\alpha_2-\beta^0_2))}.$  Then, similar to \eqref{eqn:noiseeqn1}-\eqref{eqn:nonlinear1}, we use the values of $\mu(\beta_1,\beta_2)$ for  $(\beta_1,\beta_2) = (\beta^0_1, \beta^0_2)$ and $(\beta_1,\beta_2) = (\beta^0_1, \beta^0_2 + w)$ to obtain

\begin{align}
    &\sqrt{\abs{\frac{\mu(\beta^0_1, \beta^0_2) - \sigma^2}{ \mu(\beta^0_1, \beta^0_2 + w) - \sigma^2}}} = \abs{\frac{\alpha_2 - \beta^0_2 - w}{\alpha_2 - \beta^0_2}}. \label{eqn:nonlinear3}
\end{align}
By solving the nonlinear equation \eqref{eqn:nonlinear3}, we obtain two solutions for $\alpha_2$, denoted by $\alpha_2^{(1)}$ and $\alpha_2^{(2)}$, given by
\begin{align}
\alpha_2^{(1)/(2)} = \beta^0_2 +  \frac{w}{1 \pm \sqrt{\abs{\frac{\mu(\beta^0_1, \beta^0_2) - \sigma^2}{\mu(\beta^0_1, \beta^0_2 + w) - \sigma^2}}}}. \label{eqn:alpha2solutionratio1}
\end{align}
To identify the correct solution for $\alpha_2$ of the two given in \eqref{eqn:alpha2solutionratio1}, we need the value of $\mu(\beta_1, \beta_2)$ for $(\beta_1,\beta_2) = (\beta^0_1, \beta^0_2 - w)$. Again, following the procedure from \eqref{eqn:noiseeqn1}-\eqref{eqn:alpha1solutionratio1}, by using the values of $\mu(\beta_1, \beta_2)$ for $(\beta^0_1, \beta^0_2)$ and $(\beta^0_1, \beta^0_2 - w),$ we obtain
\begin{align}
    \alpha_2^{(3)/(4)} = \beta^0_2 -  \frac{w}{1 \pm \sqrt{\abs{\frac{\mu(\beta^0_1, \beta^0_2) - \sigma^2}{\mu(\beta^0_1, \beta^0_2 - w) - \sigma^2}}}}. \label{eqn:alpha2solutionratio2}
\end{align}
 One of the solutions in \eqref{eqn:alpha2solutionratio1} is exactly same as the solutions of \eqref{eqn:alpha2solutionratio2}. Therefore, using \eqref{eqn:alpha2solutionratio1} and \eqref{eqn:alpha2solutionratio2}, the correct solution of $\alpha_2$ can be obtained as
 \begin{align}
\alpha_2 = \left\{\frac{\alpha_2^{(i)}  + \alpha_2^{(j)}}{2} : \min_{i \in \{1,2\}, j \in \{3,4\}}\abs{\alpha_2^{(i)} - \alpha_2^{(j)}}\right\}.\label{eqn:noiseestalpha2final} 
\end{align}
Finally, by setting $\beta^*_1 = \alpha_1$ and $\beta^*_2 = \alpha_2$, where $\alpha_1$ and $\alpha_2$ are given by \eqref{eqn:noiseestalpha1final} and \eqref{eqn:noiseestalpha2final}, respectively, we obtain \eqref{eqn:beta1opt} and \eqref{eqn:beta2opt}. \qedhere
\end{proof}
\begin{rem}
In \cite[Sec. IV.A]{ghermezcheshmeh2021channel}, the authors proposed the channel estimation strategy under the assumption that there is no noise in the system. However, in the noisy case, we proposed an estimation scheme based on the assumption that $\mu(\beta_1, \beta_2)$ for any of the phase-shifting parameters $(\beta_1,\beta_2) \in \mathcal{B}$ are perfectly known at the HMT.
\end{rem}
However, in practice the exact values of $\mu(\beta_1,\beta_2)$ for any of the phase-shifting parameters $(\beta_1,\beta_2) \in \mathcal{B}$ cannot be known in advance at the HMT, and therefore they need to be estimated using pilot symbols. In the following, we propose an algorithm that estimates $\mu(\beta_1,\beta_2)$ for the phase-shifting parameters in $\mathcal{B}$ and then uses the estimated values of $\mu(\beta_1,\beta_2)$ to find the optimal phase-shifting parameters $(\beta^*_1,\beta^*_2)$ in the presence of noise.

\subsection{Estimation Of The Optimal Phase-Shifting Parameters In The Noisy Case}

The user sends in total $N$ number of pilot signals to the HMT for the estimation of the five values of $\mu(\beta_1,\beta_2)$ for the five pairs of $(\beta_1,\beta_2) \in \mathcal{B}.$ As a result, the proposed algorithm works in five epochs. In the $k^{th}$ epoch, for $k=1,2,\dots$, the user transmits $\floor{\frac{N}{5}}$ number of pilots to the HMT. The HMT sets $(\beta_1,\beta_2)$ to the $k^{th}$ element in $\mathcal{B}$, and collects $\floor{\frac{N}{5}}$ samples of the received signal squared, given by \eqref{eqn:absRSS}. Then $\mu(\beta_1,\beta_2)$, for $(\beta_1,\beta_2)$ being the $k^{th}$ elements in $\mathcal{B}$, is estimated as 
\begin{align}
\label{eqn:empmeans}
    \hat{\mu}(\beta_1,\beta_2) = \frac{1}{\floor{N/5}}\sum_{i=1}^{\floor{N/5}}r_i(\beta_1,\beta_2),
\end{align}
where $r_i(\beta_1,\beta_2)$ is the $i^{th}$ sample of $r(\beta_1,\beta_2)$ in \eqref{eqn:absRSS}.

Next, we replace $\mu(\beta_1,\beta_2)$ in \eqref{eqn:alpha1solutionratio1}, \eqref{eqn:alpha1solutionratio2}, \eqref{eqn:alpha2solutionratio1}, and \eqref{eqn:alpha2solutionratio2}  by $\hat{\mu}(\beta_1,\beta_2),$ $\forall (\beta_1,\beta_2) \in \mathcal{B}$, and thereby obtain our estimates for $\beta^*_1$ and $\beta^*_2$, denoted by $\hat{\beta}^*_1$ and $\hat{\beta}^*_2$. The pseudo-code of the proposed algorithm is given in \ref{algo:twostagealgo} below. 
\begin{algorithm}[h]
	\renewcommand{\thealgorithm}{Two-Stage Phase-Shifts Estimation Algorithm}
	\floatname{algorithm}{}
	\caption{\bf }
	\label{algo:twostagealgo}
    \begin{algorithmic}[1]
        \STATE \textbf{Input:} $N, \mathcal{B}, \sigma^2.$
        \STATE \textbf{***Stage 1: Uniform Exploration ***}
        \FOR{$k= 1$ to $5$}
            \STATE HMT sets $(\beta_1,\beta_2)$ to the $k^{th}$ pair in $\mathcal{B}.$
            \STATE User sends $\floor{N/5}$ number of pilots to the HMT. 
              \STATE For the $i^{th}$ pilot, the HMT receives $r_i(\beta_1,\beta_2)$, given by \eqref{eqn:absRSS}, for $i = 1,2,\dots, \floor{N/5}.$
              \STATE The HMT computes  $\hat{\mu}_k(\beta_1,\beta_2)$ using \eqref{eqn:empmeans}. 
        \ENDFOR
        \STATE \textbf{***Stage 2: Estimate Optimal Phase-Shifting Parameters***}
        \STATE Obtain $\hat{\beta}^*_1$ as
    \begin{align}
    \label{eqn:estimatedbeta1}
        \hat{\beta}^*_1 = \left\{\frac{\hat{\alpha}_1^{(i)}  + \hat{\alpha}_1^{(j)}}{2} : \min_{i \in \{1,2\}, j \in \{3,4\}}\abs{\hat{\alpha}_1^{(i)} - \hat{\alpha}_1^{(j)}}\right\},
    \end{align}
    where $\hat{\alpha}_1^{(1)/(2)}$ is obtained by replacing the value of $\mu(\beta_1,\beta_2)$ by $\hat{\mu}(\beta_1,\beta_2)$ in \eqref{eqn:alpha1solutionratio1}, and $\hat{\alpha}_1^{(3)/(4)}$ is obtained by replacing the value of $\mu(\beta_1,\beta_2)$ by $\hat{\mu}(\beta_1,\beta_2)$ in \eqref{eqn:alpha1solutionratio2}. 
    \STATE    Obtain $\hat{\beta}^*_2$ as
    \begin{align}
    \label{eqn:estimatedbeta2}
        \hat{\beta}^*_2 = \left\{\frac{\hat{\alpha}_2^{(i)}  + \hat{\alpha}_2^{(j)}}{2} : \min_{i \in \{1,2\}, j \in \{3,4\}}\abs{\hat{\alpha}_2^{(i)} - \hat{\alpha}_2^{(j)}}\right\},
    \end{align}
    where $\hat{\alpha}_2^{(1)/(2)}$ is obtained by replacing the value of $\mu(\beta_1,\beta_2)$ by $\hat{\mu}(\beta_1,\beta_2)$ in \eqref{eqn:alpha2solutionratio1}, and $\hat{\alpha}_2^{(3)/(4)}$ is obtained by replacing the value of $\mu(\beta_1,\beta_2)$ by $\hat{\mu}(\beta_1,\beta_2)$ in \eqref{eqn:alpha2solutionratio2}.
   \STATE \textbf{Output:} $\hat{\beta}^*_1$ and $\hat{\beta}^*_2.$\\
    \STATE  \textbf{Phase-shifts at HMT} Set the phase-shift of the $(m_x,m_y)^{th}$ element at the HMT to
      $$\beta_{m_xm_y} = -\mod(k_0d_r(m_x\hat{\beta}^*_1 + m_y\hat{\beta}^*_2),2\pi).$$
    \end{algorithmic}
\end{algorithm}
We note that the choice of the initial $(\beta^0_1, \beta^0_2)$ in the set $\mathcal{B}$ was arbitrary. The values of $(\beta^0_1, \beta^0_2)$ can effect the estimation error. In general, if the values $(\beta^0_1, \beta^0_2)$ are closer to the $(\alpha_1, \alpha_2),$ the better the estimation will be. A good choice for $(\beta^0_1, \beta^0_2)$ is given in \cite[Sec. V.C]{ghermezcheshmeh2021channel}, which leads to faster learning of $(\alpha_1, \alpha_2)$.

%% file: Analysis.tex
In the section, we bound the probability that the estimates, obtained from the proposed \ref{algo:twostagealgo} Algorithm, deviate from the true values of $(\alpha_1,\alpha_2)$ by an amount $0\leq \epsilon \leq 1$. In particular, we upper bound the following error probability  
\begin{align}
\label{eqn:ErrProb}
    \Prob{\bigg(\hat{\beta}^*_1-\alpha_1\bigg)^2 + \bigg(\hat{\beta}^*_2-\alpha_2\bigg)^2 \geq \epsilon}.
\end{align}
 We use the following results to upper bound the error probability in \eqref{eqn:ErrProb}.

\begin{lem}
\label{prop:prob_ineq}
Let $\{X_n\}$ be a sequence of random variables (RVs) on a probability space. Let $X$ be a RV defined on the same probability space.
Then, the following holds $$\Prob{\abs{X_n - X_m} \geq \epsilon} \leq \Prob{\abs{X_n - X} \geq \frac{\epsilon}{2}} + \Prob{\abs{X_m - X} \geq \frac{\epsilon}{2}}.$$
\end{lem}
\begin{proof}
The proof is given in the Appendix A.
\end{proof}
Let $\chi_p^2(\lambda)$ denote a non-central Chi-squared distribution with $p$ degrees of freedom and non-centrality parameter $\lambda$. 

\begin{lem}
\label{lem:chi_dist}
Let $X=\frac{2}{\sigma^2}r(\beta_1, \beta_2)$, where $r(\beta_1,\beta_2)$ is given by \eqref{eqn:absRSS}, and let $\lambda_1= \frac{2}{\sigma^2}\abs{\sqrt{P}H(\beta_1, \beta_2)}^2$. Then, $X$ is distributed as  $\chi_{2}^2(\lambda_1)$, i.e., $X \sim \chi_{2}^2(\lambda_1).$ Furthermore, if $X_i$ for $i=1,2,\dots,n$ are $n$ independently and identically distributed (i.i.d.) RVs of $\chi_{2}^2(\lambda_1)$, then $$\sum\limits_{i=1}^n X_i \sim \chi_{2n}^2(n\lambda_1).$$
\end{lem}
\begin{proof}
The proof is given in the Appendix B.
\end{proof}

The following theorem provides an upper bound on the error probability in \eqref{eqn:ErrProb}.

\begin{thm}
\label{thm:upperboundErrProb}
Let us perform uniform exploration on the set $\mathcal{B}$ given in \eqref{eqn:setpar}. For any $0 \leq \epsilon \leq 1$, 
the error probability in \eqref{eqn:ErrProb} is upper bounded as
\begin{align}
    \mathbb{P}\bigg\{\bigg(\hat{\beta}^*_1-\alpha_1\bigg)^2 + \bigg(\hat{\beta}^*_2&-\alpha_2\bigg)^2 \geq \epsilon\bigg\} \leq  4\bigg\{ e^{-\frac{n}{32}\big(\frac{\epsilon\lambda_2}{1+\lambda_2}\big)^2} + e^{-\frac{n}{32}\big(\frac{\epsilon\lambda_3}{1+\lambda_3}\big)^2} \nonumber\\
    &\quad+ e^{-\frac{n}{32}\big(\frac{\epsilon\lambda_4}{1+\lambda_4}\big)^2} + e^{-\frac{n}{32}\big(\frac{\epsilon\lambda_5}{1+\lambda_5}\big)^2}\bigg\},\label{eqn:totalboundalpha1alpha2}
\end{align}
where
\begin{align}
    &\lambda_1 = \frac{2\abs{\sqrt{P}H(\beta^0_1, \beta^0_2)}^2}{\sigma^2},\quad  \lambda_2 = \frac{2\abs{\sqrt{P}H(\beta^0_1 + v, \beta^0_2)}^2}{\sigma^2}, \nonumber\\ 
    &\lambda_3 = \frac{2\abs{\sqrt{P}H(\beta^0_1 - v, \beta^0_2)}^2}{\sigma^2}, \quad \lambda_4 = \frac{2\abs{\sqrt{P}H(\beta^0_1, \beta^0_2 + w)}^2}{\sigma^2},\nonumber\\ 
   & \lambda_5 = \frac{2\abs{\sqrt{P}H(\beta^0_1, \beta^0_2 - w)}^2}{\sigma^2}.\nonumber
\end{align}
\end{thm}
\begin{proof}
Let us denote the estimate of $\mu(\beta^0_1, \beta^0_2)$ by $\hat{\mu}(\beta^0_1, \beta^0_2)$ which is given by
\begin{align}
    \hat{\mu}(\beta^0_1, \beta^0_2) &= \frac{1}{n}\sum\limits_{i=1}^n r_i(\beta^0_1, \beta^0_2) = \frac{\sigma^2}{2n}\sum\limits_{i=1}^n X_i. \nonumber
\end{align}
Using Lemma~\ref{lem:chi_dist}, we have
\begin{align}
    \hat{\mu}_1 &:=\frac{2n}{\sigma^2}\hat{\mu}(\beta^0_1, \beta^0_2) \sim \chi^2_{2n}(n\lambda_1) \label{eqn:empmeanbeta1beta2}\\
    \hat{\mu}_2 &:=\frac{2n}{\sigma^2}\hat{\mu}(\beta^0_1 + v, \beta^0_2) \sim \chi^2_{2n}(n\lambda_2) \label{eqn:empmeansbeta1nextbeta2}\\
    \hat{\mu}_3 &:=\frac{2n}{\sigma^2}\hat{\mu}(\beta^0_1 - v, \beta^0_2) \sim \chi^2_{2n}(n\lambda_3) \label{eqn:empmeansbeta1prevbeta2}\\
    \hat{\mu}_4 &:= \frac{2n}{\sigma^2}\hat{\mu}(\beta^0_1, \beta^0_2 + w) \sim \chi^2_{2n}(n\lambda_4) \label{eqn:empmeansbeta1beta2next} \\
     \hat{\mu}_5 &:=\frac{2n}{\sigma^2}\hat{\mu}(\beta^0_1, \beta^0_2 - w) \sim \chi^2_{2n}(n\lambda_5) \label{eqn:empmeansbeta1beta2prev}
\end{align}
where $\lambda_1, \lambda_2, \lambda_3, \lambda_4$ and $\lambda_5$ is given in Theorem \ref{thm:upperboundErrProb}.

The random variables $\hat{\mu}_1, \hat{\mu}_2, \hat{\mu}_3, \hat{\mu}_4,$ and $\hat{\mu}_5$ are mutually independent, since they are sampled at different epochs. The estimated optimal phase-shifting parameters $(\hat{\beta}^*_1,\hat{\beta}^*_2)$, are given by \eqref{eqn:estimatedbeta1} and \eqref{eqn:estimatedbeta2}, where the values of $\hat{\alpha}_1^{(1)}$, $\hat{\alpha}_1^{(2)}$, $\hat{\alpha}_1^{(3)}$, $\hat{\alpha}_1^{(4)}$, and, $\hat{\alpha}_2^{(1)}$, $\hat{\alpha}_2^{(2)}$, $\hat{\alpha}_2^{(3)}$, and $\hat{\alpha}_2^{(4)}$ are given by
\begin{align}
    \hat{\alpha}_1^{(1)/(2)} &= \beta^0_1 + \frac{v}{ 1 \pm \sqrt{\abs{\frac{\hat{\mu}\left(\beta^0_1, \beta^0_2\right) - \sigma^2}{\hat{\mu}\left(\beta^0_1 + v, \beta^0_2\right) - \sigma^2}}}} \label{eqn:hatbeta112}\\
    \hat{\alpha}_1^{(3)/(4)} &=  \beta^0_1 - \frac{v}{ 1 \pm \sqrt{\abs{\frac{\hat{\mu}\left(\beta^0_1, \beta^0_2\right) - \sigma^2}{\hat{\mu}\left(\beta^0_1 - v, \beta^0_2\right) - \sigma^2}}}}\label{eqn:hatbeta123}\\
    \hat{\alpha}_2^{(1)/(2)} &=  \beta^0_2 + \frac{w}{1 \pm \sqrt{\abs{\frac{\hat{\mu}\left(\beta^0_1, \beta^0_2\right) - \sigma^2}{\hat{\mu}\left(\beta^0_1, \beta^0_2 + w\right) - \sigma^2}}}} \label{eqn:hatbeta212}\\
    \hat{\alpha}_2^{(3)/(4)} &= \beta^0_2 - \frac{w}{1 \pm \sqrt{\abs{\frac{\hat{\mu}\left(\beta^0_1, \beta^0_2\right) - \sigma^2}{\hat{\mu}\left(\beta^0_1, \beta^0_2 - w\right) - \sigma^2}}}}.\label{eqn:hatbeta223}
\end{align}
By inserting \eqref{eqn:empmeanbeta1beta2}, \eqref{eqn:empmeansbeta1nextbeta2}, \eqref{eqn:empmeansbeta1prevbeta2}, \eqref{eqn:empmeansbeta1beta2next}, and \eqref{eqn:empmeansbeta1beta2prev} into \eqref{eqn:hatbeta112}, \eqref{eqn:hatbeta123}, \eqref{eqn:hatbeta212} and \eqref{eqn:hatbeta223}, we obtain 
\begin{align}
    &\hat{\alpha}_1^{(1)/(2)} = \beta^0_1 + \frac{v}{ 1 \pm \sqrt{\abs{\frac{\hat{\mu}_1 - 2n}{\hat{\mu}_2 - 2n}}}}  \label{eqn:alpha1hat12}\\
    &\hat{\alpha}_1^{(3)/(4)} = \beta^0_1 - \frac{v}{ 1 \pm \sqrt{\abs{\frac{\hat{\mu}_1 - 2n}{\hat{\mu}_3 - 2n}}}} \label{eqn:alpha1hat34}\\
    &\hat{\alpha}_2^{(1)/(2)}   = \beta^0_2 + \frac{w}{ 1 \pm \sqrt{\abs{\frac{\hat{\mu}_1 - 2n}{\hat{\mu}_4 - 2n}}}}  \label{eqn:alpha2hat12} \\
    &\hat{\alpha}_2^{(3)/(4)}  =  \beta^0_2 - \frac{w}{ 1 \pm \sqrt{\abs{\frac{\hat{\mu}_1 - 2n}{\hat{\mu}_5 - 2n}}}}.\label{eqn:alpha2hat34}
\end{align}
Let us denote
\begin{align}
    I:= \Prob{\abs{\hat{\beta}^*_1-\alpha_1} \geq \sqrt{\frac{\epsilon}{2}}}\nonumber\\
    II:= \Prob{\abs{\hat{\beta}^*_2-\alpha_2} \geq \sqrt{\frac{\epsilon}{2}}}.\nonumber
\end{align}
Now, applying Lemma \ref{prop:prob_ineq} in \eqref{eqn:ErrProb}, we obtain
\begin{align}
    \mathbb{P}\bigg\{&\bigg(\hat{\beta}^*_1-\alpha_1\bigg)^2 + \bigg(\hat{\beta}^*_2-\alpha_2\bigg)^2 \geq \epsilon\bigg\}\nonumber\\
    &\leq \Prob{\bigg(\hat{\beta}^*_1-\alpha_1\bigg)^2 \geq \frac{\epsilon}{2}} + \Prob{\bigg(\hat{\beta}^*_2-\alpha_2\bigg)^2 \geq \frac{\epsilon}{2}}\nonumber\\
   & \leq I + II. \label{eqn:totalprob}
\end{align}
We upper bound each of the term in right-hand side of \eqref{eqn:totalprob}. We begin with the first term $\Prob{\abs{\hat{\beta}^*_1-\alpha_1} \geq \sqrt{\frac{\epsilon}{2}}}$, denoted as I.

\begin{itemize}
\item[\ding{108}] \textbf{Step 1: Upper bound on I}

From \eqref{eqn:estimatedbeta1}, we have

\begin{align}
    &\Prob{\abs{\hat{\beta}^*_1 - \alpha_1} \geq  \sqrt{\frac{\epsilon}{2}}} \nonumber\\
    &= \mathbb{P}\Bigg\{\left\{\abs{\frac{\hat{\alpha}_1^{(1)} + \hat{\alpha}_1^{(3)}}{2} - \alpha_1}\geq  \sqrt{\frac{\epsilon}{2}} \right\} \nonumber\\
    &\quad\bigcup \left\{\abs{\frac{\hat{\alpha}_1^{(1)} + \hat{\alpha}_1^{(4)}}{2} - \alpha_1}\geq  \sqrt{\frac{\epsilon}{2}} \right\} \nonumber\\
    &\quad\bigcup \left\{\abs{\frac{\hat{\alpha}_1^{(2)} + \hat{\alpha}_1^{(3)}}{2} - \alpha_1}\geq  \sqrt{\frac{\epsilon}{2}} \right\} \nonumber\\
    &\quad\bigcup \left\{\abs{\frac{\hat{\alpha}_1^{(2)} + \hat{\alpha}_1^{(4)}}{2} - \alpha_1}\geq  \sqrt{\frac{\epsilon}{2}} \right\}\Bigg\}\nonumber\\ 
  &\leq \mathbb{P}\bigg\{\abs{\frac{\hat{\alpha}_1^{(1)} + \hat{\alpha}_1^{(3)}}{2} - \alpha_1}\geq  \sqrt{\frac{\epsilon}{2}}\bigg\} \nonumber\\
  &\quad + \mathbb{P}\bigg\{\abs{\frac{\hat{\alpha}_1^{(1)} + \hat{\alpha}_1^{(4)}}{2} - \alpha_1}\geq  \sqrt{\frac{\epsilon}{2}}\bigg\} \nonumber\\
  &\quad + \mathbb{P}\bigg\{\abs{\frac{\hat{\alpha}_1^{(2)} + \hat{\alpha}_1^{(3)}}{2} - \alpha_1}\geq  \sqrt{\frac{\epsilon}{2}}\bigg\} \nonumber\\
  &\quad+ \mathbb{P}\bigg\{\abs{\frac{\hat{\alpha}_1^{(2)} + \hat{\alpha}_1^{(4)}}{2} - \alpha_1}\geq  \sqrt{\frac{\epsilon}{2}}\bigg\}\nonumber\\
    &= \sum_{\substack{i=1,2\\j=3,4}}\Prob{\abs{\bigg(\hat{\alpha}_1^{(i)} - \alpha_1\bigg) + \bigg(\hat{\alpha}_1^{(j)} - \alpha_1\bigg)}\geq  2\sqrt{\frac{\epsilon}{2}}} \nonumber\\
    &\leq \sum_{\substack{i=1,2\\j=3,4}}\bigg[\Prob{\abs{\hat{\alpha}_1^{(i)} - \alpha_1} \geq \sqrt{\frac{\epsilon}{2}}}  + \Prob{\abs{\hat{\alpha}_1^{(j)} - \alpha_1}\geq  \sqrt{\frac{\epsilon}{2}}} \bigg]\nonumber\\
    &= 2\bigg(\Prob{\abs{\hat{\alpha}_1^{(1)} - \alpha_1} \geq  \sqrt{\frac{\epsilon}{2}}}  + \Prob{\abs{\hat{\alpha}_1^{(2)} - \alpha_1} \geq  \sqrt{\frac{\epsilon}{2}}} \nonumber\\
    &\quad + \Prob{\abs{\hat{\alpha}_1^{(3)} - \alpha_1} \geq  \sqrt{\frac{\epsilon}{2}}} + \Prob{\abs{\hat{\alpha}_1^{(4)} - \alpha_1} \geq  \sqrt{\frac{\epsilon}{2}}}\bigg)\label{eqn:(2)},
\end{align}
where we applied the union bound to get the first inequality and applied Lemma~\ref{prop:prob_ineq} for the second inequality. We now bound each term in \eqref{eqn:(2)} separately.

\begin{itemize}
    \item[\ding{228}] \textbf{Upper bound of $\boldsymbol{\Prob{\abs{\hat{\alpha}_1^{(1)} - \alpha_1} \geq  \sqrt{\frac{\epsilon}{2}}}}$:} Substituting the values of $\hat{\alpha}_1^{(1)}$, as given by \eqref{eqn:alpha1hat12}, in $\boldsymbol{\Prob{\abs{\hat{\alpha}_1^{(1)} - \alpha_1} \geq  \sqrt{\frac{\epsilon}{2}}}}$,  we obtain
    \begin{align}
&\Prob{\abs{\hat{\alpha}_1^{(1)} - \alpha_1)} \geq  \sqrt{\frac{\epsilon}{2}}}  \nonumber\\
    &= \Prob{\abs{\frac{v}{ 1 + \sqrt{\abs{\frac{\hat{\mu}_1 - 2n}{\hat{\mu}_2 - 2n}}}} - (\alpha_1 - \beta^0_1)} \geq  \sqrt{\frac{\epsilon}{2}}}. 
   \label{eqn:hatbeta1ub1}
   \end{align}
  Note that the following holds.
    \begin{align}
     \abs{\frac{v}{ 1 + \sqrt{\abs{\frac{\hat{\mu}_1 - 2n}{\hat{\mu}_2 - 2n}}}} - (\alpha_1 - \beta^0_1)} &\leq \abs{\frac{v}{ 1 + \sqrt{\abs{\frac{\hat{\mu}_1 - 2n}{\hat{\mu}_2 - 2n}}}}} +\Bigg|\alpha_1 - \beta^0_1\Bigg|. \label{eqn:mod_bound}
     \end{align}
 By applying \eqref{eqn:mod_bound} in \eqref{eqn:hatbeta1ub1}, we obtain
  \begin{align}
    &\Prob{\abs{\hat{\alpha}_1^{(1)} - \alpha_1)} \geq  \sqrt{\frac{\epsilon}{2}}} \nonumber\\
      & \leq \Prob{\abs{\frac{1}{1 + \sqrt{\abs{\frac{\hat{\mu}_1 - 2n}{\hat{\mu}_2 - 2n}}}}} \geq  \frac{1}{v}\left( \sqrt{\frac{\epsilon}{2}}- \Bigg|\alpha_1 - \beta^0_1\Bigg|\right)}.\label{eqn:hatbeta1ub2}
      \end{align}
    For the RVs $\hat{\mu}_1$ and $\hat{\mu}_2$, $\frac{1}{1 + \sqrt{\abs{\frac{\hat{\mu}_1 - 2n}{\hat{\mu}_2 - 2n}}}}$ is always positive. Using this fact in \eqref{eqn:hatbeta1ub2}, we obtain
\begin{align}
  &\Prob{\abs{\hat{\alpha}_1^{(1)} - \alpha_1)} \geq  \sqrt{\frac{\epsilon}{2}}}      \nonumber\\
  &\leq \Prob{\frac{1}{1 + \sqrt{\abs{\frac{\hat{\mu}_1 - 2n}{\hat{\mu}_2 - 2n}}}} \geq  \frac{1}{v}\left( \sqrt{\frac{\epsilon}{2}}- \Bigg|\alpha_1 - \beta^0_1\Bigg|\right)}.\nonumber
\end{align}
Let $a = \frac{1}{v}\left( \sqrt{\frac{\epsilon}{2}}- \abs{\alpha_1 - \beta^0_1}\right).$ We have
\begin{align}
 \Prob{\abs{\hat{\alpha}_1^{(1)} - \alpha_1)} \geq  \sqrt{\frac{\epsilon}{2}}} &\leq \Prob{1 + \sqrt{\abs{\frac{\hat{\mu}_1 - 2n}{\hat{\mu}_2 - 2n}}} \leq \frac{1}{a}}\nonumber\\
&= \Prob{\abs{\frac{\hat{\mu}_1 - 2n}{\hat{\mu}_2 - 2n}} \leq \left(1 - \frac{1}{a}\right)^2}. \label{eqn:boundhatalpha11}   
\end{align}
    
    \item[\ding{228}] \textbf{Upper bound of $\boldsymbol{\Prob{\abs{\hat{\alpha}_1^{(2)} - \alpha_1} \geq  \sqrt{\frac{\epsilon}{2}}}}$:} Substituting the values of $\hat{\alpha}_1^{(2)}$, as given by \eqref{eqn:alpha1hat12}, in $\boldsymbol{\Prob{\abs{\hat{\alpha}_1^{(2)} - \alpha_1} \geq  \sqrt{\frac{\epsilon}{2}}}}$, we obtain
    \begin{align}
        &\Prob{\abs{\hat{\alpha}_1^{(2)} - \alpha_1)} \geq  \sqrt{\frac{\epsilon}{2}}}  \nonumber\\
        &= \Prob{\abs{\frac{v}{ 1 - \sqrt{\abs{\frac{\hat{\mu}_1 - 2n}{\hat{\mu}_2 - 2n}}}} - (\alpha_1 - \beta^0_1)} \geq  \sqrt{\frac{\epsilon}{2}}}.\label{eqn:estalpha12}
     \end{align}   
    Note that the following holds.
     \begin{align}
       \abs{\frac{v}{ 1 - \sqrt{\abs{\frac{\hat{\mu}_1 - 2n}{\hat{\mu}_2 - 2n}}}} - (\alpha_1 - \beta^0_1)} &\leq \abs{\frac{v}{ 1 - \sqrt{\abs{\frac{\hat{\mu}_1 - 2n}{\hat{\mu}_2 - 2n}}}}} +\Bigg|\alpha_1 - \beta^0_1\Bigg|. \label{eqn:mod_bound2}  
      \end{align}
     By applying \eqref{eqn:mod_bound2} in the right-hand side of \eqref{eqn:estalpha12}, we obtain
     \begin{align}
    &\Prob{\abs{\hat{\alpha}_1^{(2)} - \alpha_1)} \geq  \sqrt{\frac{\epsilon}{2}}}  \nonumber\\
     &\leq \Prob{\abs{1 - \sqrt{\abs{\frac{\hat{\mu}_1 - 2n}{\hat{\mu}_2 - 2n}} }} \leq \frac{1}{a}}\nonumber\\
     &= \Prob{\left(1 - \frac{1}{a}\right)^2 \leq \abs{\frac{\hat{\mu}_1 - 2n}{\hat{\mu}_2 - 2n}} \leq \left(1 + \frac{1}{a}\right)^2}\nonumber\\
     &\leq \Prob{\abs{\frac{\hat{\mu}_1 - 2n}{\hat{\mu}_2 - 2n}} \leq \left(1 + \frac{1}{a}\right)^2} - \Prob{\abs{\frac{\hat{\mu}_1 - 2n}{\hat{\mu}_2 - 2n}} \leq \left(1 - \frac{1}{a}\right)^2}.  \label{eqn:boundhatalpha12}
    \end{align}
    
    \item[\ding{228}] \textbf{Upper bound of $\boldsymbol{\Prob{\abs{\hat{\alpha}_1^{(3)} - \alpha_1} \geq  \sqrt{\frac{\epsilon}{2}}}}$:} Substituting the values of $\hat{\alpha}_1^{(3)}$, as given by \eqref{eqn:alpha1hat34}, in $\boldsymbol{\Prob{\abs{\hat{\alpha}_1^{(3)} - \alpha_1} \geq  \sqrt{\frac{\epsilon}{2}}}}$ and following similar steps to bound $\boldsymbol{\Prob{\abs{\hat{\alpha}_1^{(1)} - \alpha_1} \geq  \sqrt{\frac{\epsilon}{2}}}}$, we obtain 
    \begin{align}
          \Prob{\abs{\hat{\alpha}_1^{(3)} - \alpha_1)} \geq  \sqrt{\frac{\epsilon}{2}}} &\leq 
    \Prob{\abs{\frac{\hat{\mu}_1 - 2n}{\hat{\mu}_3 - 2n}} \leq \left(1 - \frac{1}{a}\right)^2}. \label{eqn:boundhatalpha13}
    \end{align}
    
    \item[\ding{228}] \textbf{Upper bound of $\boldsymbol{\Prob{\abs{\hat{\alpha}_1^{(4)} - \alpha_1} \geq  \sqrt{\frac{\epsilon}{2}}}}$:} Substituting the values of $\hat{\alpha}_1^{(4)}$, as given by \eqref{eqn:alpha1hat34}, in $\boldsymbol{\Prob{\abs{\hat{\alpha}_1^{(4)} - \alpha_1} \geq  \sqrt{\frac{\epsilon}{2}}}}$ and following similar steps to bound $\boldsymbol{\Prob{\abs{\hat{\alpha}_2^{(2)} - \alpha_1} \geq  \sqrt{\frac{\epsilon}{2}}}}$, we obtain 
    \begin{align}
         &\Prob{\abs{\hat{\alpha}_1^{(4)} - \alpha_1)} \geq  \sqrt{\frac{\epsilon}{2}}} \nonumber\\
         &\leq 
    \Prob{\abs{\frac{\hat{\mu}_1 - 2n}{\hat{\mu}_3 - 2n}} \leq \left(1 + \frac{1}{a}\right)^2} - \Prob{\abs{\frac{\hat{\mu}_1 - 2n}{\hat{\mu}_3 - 2n}} \leq \left(1 - \frac{1}{a}\right)^2}.  \label{eqn:boundhatalpha14} 
    \end{align}
\end{itemize}
By inserting the bounds \eqref{eqn:boundhatalpha11}, \eqref{eqn:boundhatalpha12}, \eqref{eqn:boundhatalpha13} and \eqref{eqn:boundhatalpha14} to \eqref{eqn:(2)} we obtain
\begin{align}
     &\Prob{\abs{\hat{\beta}^*_1 - \alpha_1} \geq  \sqrt{\frac{\epsilon}{2}}}\nonumber\\
    &\leq 2\left(\Prob{\abs{\frac{\hat{\mu}_2 - 2n}{\hat{\mu}_1 - 2n}} \geq \gamma_1} + \Prob{\abs{\frac{\hat{\mu}_3 - 2n}{\hat{\mu}_1 - 2n}} \geq \gamma_1}\right), \label{eqn:boundtotalhatalpha1}
\end{align}
where we set $\gamma_1 = \left(\frac{1}{1+(1/a)}\right)^2.$

We next upper bound each term on the right-hand side of \eqref{eqn:boundtotalhatalpha1}. The bounds are derived using the properties of the sub-exponential distributions which we introduce below.

\item[\ding{108}] \textbf{Step 2: Sub-exponential Distributions and its Tail Bound}
\begin{defn}[sub-exponential distribution]
\label{defn:Subexponential}
 A RV $X$ with mean $\mu$ is said to be sub-exponential with parameters $(\nu,\alpha)$, for $\alpha > 0$, if    
  \[\EE{\exp\bigg(t(X-\mu)\bigg)}\leq \exp\left(\frac{t^2\nu^2}{2}\right), \text{ for } |t| < \frac{1}{\alpha}.\] 
\end{defn} 

\begin{thm}[\cite{wainwright2019high}]
\label{thm:sub_exp_tail_bound}
Let $X_k$ for $k=1,2,\dots,n$ be independent RVs where $X_k$ is sub-exponential with parameters $(\nu_k,b_k),$ and mean $\mu_k= \EE{X_k}$. 
Then $\sum\limits_{k=1}^n(X_k - \mu_k)$ is a sub-exponential RV with parameters $(\nu_{*},b_{*})$ where
\[b_{*} = \max\limits_{k=1,2\dots,n} b_k ,\]
 and \[ \nu_{*} = \sqrt{\sum\limits_{k=1}^n \nu_k^2}.\]
Furthermore, its tail probability can be bounded as
\begin{align}
    &\Prob{\abs{\frac{1}{n}\sum\limits_{k=1}^n(X_k - \mu_k)} \geq t} \leq
    \begin{cases}
    2e^{-\frac{nt^2}{2(\nu_{*}^2/n)}}, &\text{for } 0\leq t \leq \frac{\nu_{*}^2}{nb_{*}}\\
    2e^{-\frac{nt}{2b_{*}}}, &\text{for } t \geq \frac{\nu_{*}^2}{nb_{*}}.\nonumber
    \end{cases}
\end{align}
\end{thm}
\begin{proof}
The proof is given in Appendix C.
\end{proof}

\begin{cor}
\label{cor:sub_exp_tail_bound}
Let $X_k$ for $k=1,2\dots,n$ be i.i.d. sub-exponential RVs with parameters $(2(2+2a), 4)$ each with mean $2+a$. Then,
\begin{align}
\Prob{\abs{\frac{1}{n}\sum\limits_{k=1}^n(X_k - \mu_k)} \geq t}  
\leq 2e^{-\frac{n t^2}{8(2+2a)^2}}, \qquad \text{for } t > 0.\nonumber
\end{align}
\end{cor}

\begin{proof}
The proof is given in then Appendix D.
\end{proof}

We use Corollary \ref{cor:sub_exp_tail_bound} to upper bound of the right-hand side terms in   \eqref{eqn:boundtotalhatalpha1}. The following lemma establishes the connection between the non-central chi-squared distribution and the sub-exponential distributions.
\begin{lem}
\label{lem:XSubExp}
Let $X\sim\chi_p^2(a)$. Then, $X$ is sub-exponential with parameters $\big(2(p+2a), 4\big).$ 
\end{lem}
\begin{proof}
The proof is given in then Appendix E.
\end{proof}

\item[\ding{108}] \textbf{Step 3: Upper Bounding Eq. \eqref{eqn:boundtotalhatalpha1}}

\begin{itemize}
    \item Recall that $\hat{\mu}_1\sim \chi_{2n}^2(n\lambda_1)$ and $\hat{\mu}_2\sim\chi_{2n}^2(n\lambda_2)$. Let $f_{\hat{\mu}_1}$ denote the pdf of $\hat{\mu}_1$. We upper bound the term $ \Prob{\abs{\frac{\hat{\mu}_2 - 2n}{\hat{\mu}_1 - 2n}} \geq \gamma_1}$ as follows
    \begin{align}
     &\Prob{\abs{\frac{\hat{\mu}_2 - 2n}{\hat{\mu}_1 - 2n}} \geq \gamma_1} \nonumber\\
     &= \int\limits_{0}^{\infty}\Prob{\bigg|\hat{\mu}_2-2n\bigg| \geq \gamma_1\bigg|u-2n\bigg|}f_{\hat{\mu}_1}(u)du\nonumber\\
     &= \int\limits_{0}^{\infty}\Prob{\bigg|\hat{\mu}_2 - 2n - n\lambda_2 + n\lambda_2\bigg|\geq \gamma_1\bigg|u-2n\bigg|}f_{\hat{\mu}_1}(u)du\nonumber\\
    & \leq \int\limits_{0}^{\infty}\Prob{\frac{1}{n}\bigg|\hat{\mu}_2-n(2+\lambda_2)\bigg| \geq \frac{\gamma_1\abs{u-2n} - n\lambda_2}{n}}f_{\hat{\mu}_1}(u)du \label{eqn:probYbound1}
\end{align}

Note that, if $\frac{\gamma_1\abs{u-2n} - n\lambda_2}{n} < 0$, then $\Prob{\abs{\frac{\hat{\mu}_2 - 2n}{\hat{\mu}_1 - 2n}} \geq \gamma_1} \leq 1$ as $\Prob{\frac{1}{n}\bigg|\hat{\mu}_2-n(2+\lambda_2)\bigg| \geq \frac{\gamma_1\abs{u-2n} - n\lambda_2}{n}} = 1$, which is trivial. 

For $\frac{\gamma_1\abs{u-2n} - n\lambda_2}{n} \geq 0$, using the assumption $0 \leq \epsilon \leq 1$ in \eqref{eqn:probYbound1}, we have
\begin{align}
    &\Prob{\abs{\frac{\hat{\mu}_2 - 2n}{\hat{\mu}_1 - 2n}} \geq \gamma_1} \nonumber\\
    &\leq \int\limits_{0}^{\infty}\mathbb{P}\bigg\{\frac{1}{n}\bigg|\hat{\mu}_2-n(2+\lambda_2)\bigg| \geq \epsilon \left(\frac{\gamma_1\abs{u-2n} - n\lambda_2}{n}\right)\bigg\}\nonumber\\
    &\qquad \qquad \times f_{\hat{\mu}_1}(u)du. \label{eqn:newprobYbound1}
\end{align}
The last inequality follows from Lemma~\ref{prop:prob_ineq}. Let $t_1:=t_1(u) =\epsilon \left(\frac{\gamma_1\abs{u-2n} - n\lambda_2}{n}\right)$. As $\EE{\hat{\mu}_2}=2n+n\lambda_2$, by applying Corollary~\ref{cor:sub_exp_tail_bound}, we obtain
\begin{align}
\Prob{\frac{1}{n}\bigg|\hat{\mu}_2-n(2+\lambda_2)\bigg| \geq t_1}
 & \leq    2e^{-\frac{nt_1^2}{8(2+2\lambda_2)^2}}, \qquad t_1 \geq 0.  
 \label{eqn:corY}
\end{align}

By applying \eqref{eqn:corY} to \eqref{eqn:newprobYbound1}, we obtain
\begin{align}
  &\Prob{\bigg|\hat{\mu}_2-2n\bigg| \geq \gamma_1\bigg|\hat{\mu}_1-2n\bigg|} \nonumber\\
  &\leq \int\limits_{0}^{\infty}2e^{-\frac{nt_1^2}{8(2+2\lambda_2)^2}}f_{\hat{\mu}_1}(u)du, \nonumber\\
  &= \int\limits_{0}^{2n}2e^{-\frac{n\bigg(\frac{\epsilon}{n} \big(\gamma_1(2n - u) - n\lambda_2\big)\bigg)^2}{8(2+2\lambda_2)^2}}f_{\hat{\mu}_1}(u)du \nonumber\\
  &\quad + \int\limits_{2n}^{\infty}2e^{-\frac{n\bigg(\frac{\epsilon}{n} \big(\gamma_1(u-2n) - n\lambda_2\big)\bigg)^2}{8(2+2\lambda_2)^2}}f_{\hat{\mu}_1}(u)du.\label{eqn:intbound1}
  \end{align}

 For $0 \leq u \leq 2n$, we have
  \begin{align}
  2e^{-\frac{n\bigg(\frac{\epsilon}{n} \big(\gamma_1(2n - u) - n\lambda_2\big)\bigg)^2}{8(2+2\lambda_2)^2}} \leq 2e^{-\frac{n\big(\epsilon\lambda_2\big)^2}{8(2+2\lambda_2)^2}}.\label{eqn:intuppbound1}
  \end{align}
  For $2n \leq u \leq \infty$, we have
  \begin{align}
2e^{-\frac{n\bigg(\frac{\epsilon}{n} \big(\gamma_1(u-2n) - n\lambda_2\big)\bigg)^2}{8(2+2\lambda_2)^2}} \leq 2e^{-\frac{n\big(\epsilon\lambda_2\big)^2}{8(2+2\lambda_2)^2}}.\label{eqn:intuppbound2}
 \end{align}
Using \eqref{eqn:intuppbound1} and \eqref{eqn:intuppbound2} in \eqref{eqn:intbound1}, we obtain
 \begin{align}
   &\Prob{\bigg|\hat{\mu}_2-2n\bigg| \geq \gamma_1\bigg|\hat{\mu}_1-2n\bigg|}   \nonumber\\
   &\leq 2e^{-\frac{n\big(\epsilon\lambda_2\big)^2}{8(2+2\lambda_2)^2}}\Prob{0<\hat{\mu}_1<2n} \nonumber\\
   & + 2e^{-\frac{n\big(\epsilon\lambda_2\big)^2}{8(2+2\lambda_2)^2}}\Prob{\hat{\mu}_1>2n}\nonumber\\
  &\Prob{\bigg|\hat{\mu}_2-2n\bigg| \geq \gamma_1\bigg|\hat{\mu}_1-2n\bigg|} \leq 2e^{-\frac{n\big(\epsilon\lambda_2\big)^2}{8(2+2\lambda_2)^2}}.\label{eqn:boundY}
\end{align}

\item  We next upper bound $ \Prob{\abs{\frac{\hat{\mu}_3 - 2n}{\hat{\mu}_1 - 2n}} \geq \gamma_1}$. 
Set $t_2 = \epsilon \left(\frac{\gamma_1\abs{u-2n} - n\lambda_3}{n}\right)$. Recall that $\hat{\mu}_3\sim \chi_{2n}^2(n\lambda_3)$. Following the steps similar to the derivation of the bound in \eqref{eqn:boundY}, we obtain
\begin{align}
 \Prob{\bigg|\hat{\mu}_3-2n\bigg| \geq \gamma_1\bigg|\hat{\mu}_1-2n\bigg|} &\leq
2e^{-\frac{n\big(\epsilon\lambda_3\big)^2}{8(2+2\lambda_3)^2}}.\label{eqn:boundZ}
 \end{align}
 \end{itemize}

Combining \eqref{eqn:boundY} and \eqref{eqn:boundZ} we obtain the following upper bound on \eqref{eqn:boundtotalhatalpha1}
\begin{align}
    \Prob{\abs{\hat{\beta}^*_1 - \alpha_1} \geq \sqrt{\frac{\epsilon}{2}}} &\leq 4\bigg( e^{-\frac{n}{32}\big(\frac{\epsilon\lambda_2}{1+\lambda_2}\big)^2} + e^{-\frac{n}{32}\big(\frac{\epsilon\lambda_3}{1+\lambda_3}\big)^2}\bigg). \label{eqn:Ibound}
\end{align}

\item[\ding{108}] \textbf{Step 4: Upper bound on II}

By following the same steps for deriving the upper bound of $\Prob{\abs{\hat{\beta}^*_1-\alpha_1} \geq \sqrt{\frac{\epsilon}{2}}}$, 
we can obtain the following bound 
\begin{align}
    \Prob{\abs{\hat{\beta}^*_2 - \alpha_2} \geq \sqrt{\frac{\epsilon}{2}}} &\leq  4\bigg( e^{-\frac{n}{32}\big(\frac{\epsilon\lambda_4}{1+\lambda_4}\big)^2} + e^{-\frac{n}{32}\big(\frac{\epsilon\lambda_5}{1+\lambda_5}\big)^2}\bigg).\label{eqn:IIbound}
\end{align}
Combining \eqref{eqn:Ibound} and \eqref{eqn:IIbound}, we obtain the required upper bound in \eqref{eqn:totalboundalpha1alpha2}.\qedhere
\end{itemize}
\end{proof}

%% file: Numericals.tex
We estimate the initial value of $(\beta^0_1, \beta^0_2)$ as given in \cite[Sec. V.C]{ghermezcheshmeh2021channel}. Based on the the initial value of $(\beta^0_1, \beta^0_2)$ we set $\mathcal{B}$ as given in \eqref{eqn:setpar},
where $v$ and $w$ are selected such that $K_xv \in \mathbb{N}$ and $K_yw \in \mathbb{N}$ respectively. In addition to the LoS path, we assume that there are 4 NLoS path components due to scatters between the user and the HMT. The elevation and azimuth angles of each NLoS path from these scatters to the center of HMT follow the uniform distribution, i.e., $U(0, 2\pi)$. Moreover, we consider the path coefficient of each NLoS path as a complex Gaussian distribution, i.e., $CN(0,\sigma_s^2)$, where $\sigma_s^2$ is 20 dB weaker than the power of the LoS component \cite{TWC2021BTformmwaveIRSwang}. The system parameters for numerical simulations are listed in Table \ref{Table 1:}.

\begin{table}[ht]
\caption{A list of system parameters for numerical simulations}
\label{Table 1:}
\centering
\begin{tabular}{|p{2cm}|p{1.5cm}|p{4.2cm}|}
\hline
\toprule
    \textbf{Parameters} &\textbf{Values} &\textbf{Description}\\
    \hline
    \midrule
    \small{$f_c$} & \small{30 GHz} & \small{Carrier frequency}\\ \hline
    \small{$\lambda$} &\small{1 cm} &\small{Wavelength}\\ \hline
    \small{$L_x$} &\small{1 m} &\small{Width of the HMT}\\ \hline
    \small{$L_y$} &\small{1 m} &\small{Length of the HMT}\\ \hline
    \small{$d_r$} &\small{$\lambda/4$} &\small{Unit element spacing}\\ \hline
    \small{$L_e$} &\small{$d_r$} &\small{Width and length of each phase-shifting element}\\ \hline
    \small{$P$} &\small{20 dBm} &\small{Transmission power of the HMT during data transmission}\\ \hline
    \small{$\sigma^2$} &\small{-115 dBm} &\small{Noise power for 200 KHz \footnotemark}\\\hline
\end{tabular}
\end{table}
\vspace{-4mm}

\subsection{Comparison Between the Proposed Algorithm and Benchmark Scheme}
According to the approximated channel model, where the phase-shift parameters at the HMT are given by $\beta_1$ and $\beta_2$, the achieved data rate at the user of the HMT-assisted wireless communication system is given by
\begin{align}
   \label{eqn:achievablerates}
    R(\beta_1,\beta_2) &= log_2\left(1 +
 \frac{P\abs{H(\beta_1,\beta_2)}^2}{\sigma^2}\right),
\end{align}
where $P$ is the transmission power at the HMT. The HMT uses the acquired CSI during the channel estimation period to maximize the received data rate by the user. Hence, we consider the achieved data rate by the user, using the acquired CSI as a performance metric. We applied the proposed algorithm in two different cases when the distance between the user and the center of the HMT $(d_0 =200$ m and when $d_0 = 10$ m. We compared our proposed algorithm with two benchmarks, the proposed algorithm in \cite{ghermezcheshmeh2021channel} and the oracle scheme where $\alpha_1$ and $\alpha_2$ are estimated perfectly and thereby the maximum rate is achieved.

In Fig. \ref{fig:Rate_vs_Algo_with_initial_choice}, we compared the achievable rates, given by \eqref{eqn:achievablerates}, of the proposed scheme and the benchmark schemes. We considered both $d_0 = 200$ m and $d_0 = 10$ m regions of the HMT with respect to the transmit power of the pilot signals when the number of the pilot signals is fixed to 23. For all the algorithms we use the same number of pilots, i.e. 23, for both the cases. Our proposed algorithm uses four pilots in each epoch and there are five epochs, which makes the total number of pilots equals to 20. We require additional three number of pilots to estimate $(\beta^0_1, \beta^0_2).$ We run the simulation for 1000 times. We see that in both cases, the proposed \ref{algo:twostagealgo} gives higher rates than other two benchmark schemes.
\begin{figure}[ht]
\vspace{-5mm}
	\centering	\includegraphics[scale = 0.55]{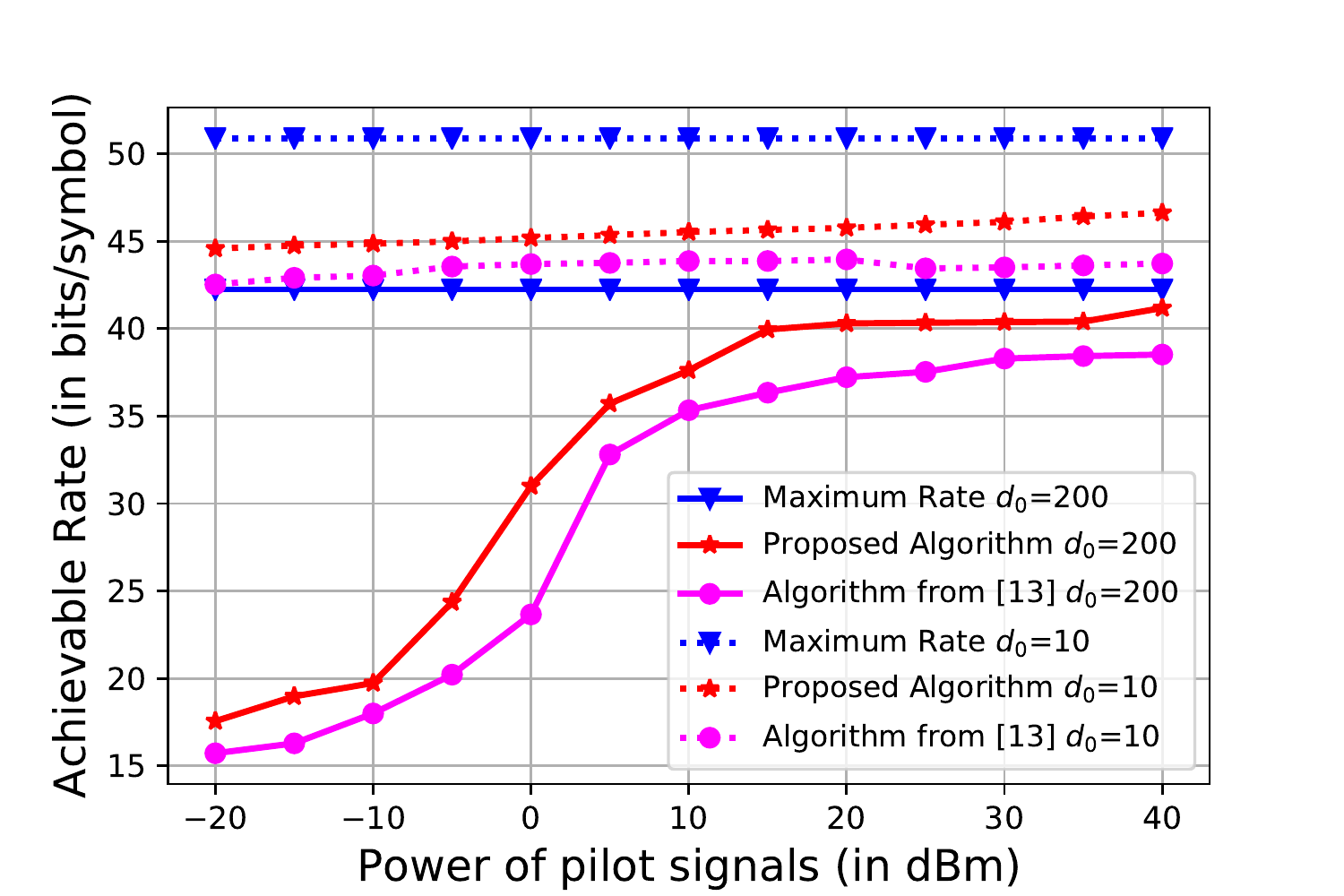}
	\caption{\small{Achievable rate vs. the transmit power of the pilot signals (in dBm).}}
	\label{fig:Rate_vs_Algo_with_initial_choice}
	\vspace{-5mm}
\end{figure}

\footnotetext{This setting corresponds to the noise power spectrum density at the HMT is $-174$ dBm/Hz and signal bandwidth is $200$ KHz, assuming the noise figure of each user to be 6 dB \cite{zhang2022beam}.}

\subsection{Convergence of The Proposed Algorithm}
We now numerically evaluate the convergence of the upper bound of the proposed algorithm, given by \eqref{eqn:totalboundalpha1alpha2}. We also compare the actual probability, given by \eqref{eqn:ErrProb}, that we obtain by simulations.

In Fig. \ref{fig:Probability_Bound_alpha1_samples_epsilon_0.02_0.9}, we show the convergence property of the error probability and its upper bound of the proposed algorithm for increasing values of $\epsilon = \{0.01, 0.05, 0.1\}$ when the power of the pilot signal is $P = 10$ dBm and $d_0 = 200$ m. We run the simulation for 1000 times. We see that for each value of $\epsilon$, the proposed algorithm converges towards zero as we increase the number of pilots. Moreover, the upper bound of the error probability also converges to the error probability as the number of pilot signals increases. 
\begin{figure}[ht]
\vspace{-5mm}
\centering
 	\includegraphics[scale = 0.6]{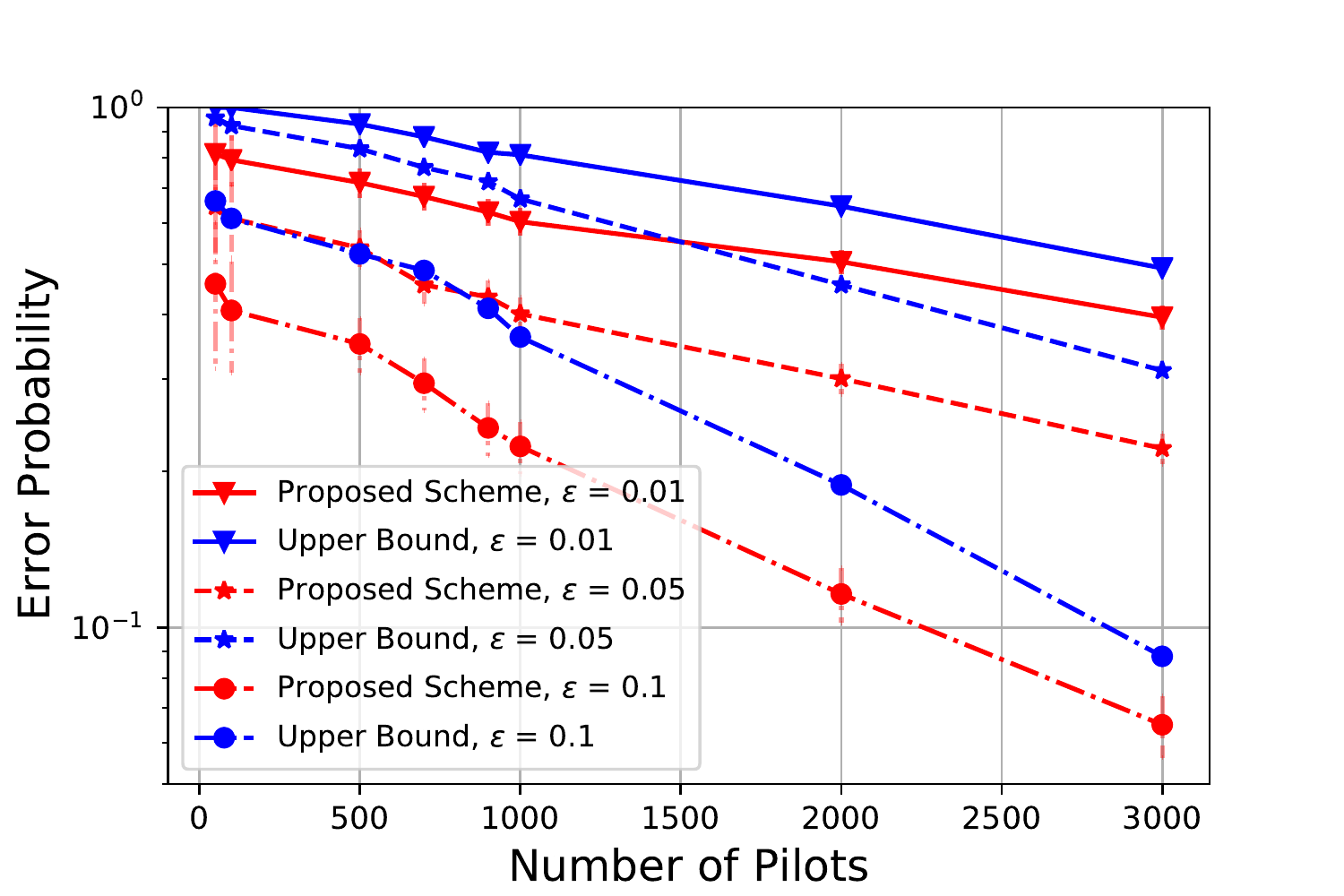}
\caption{\small{Error Probability Bound v/s Number of Pilots for $\epsilon = \{ 0.01, 0.05, 0.1\}$ for $P = 10$ dBm.}}
	\label{fig:Probability_Bound_alpha1_samples_epsilon_0.02_0.9}
	\vspace{-3mm}
\end{figure}

In Fig. \ref{fig:Probability_Bound_alpha1_samples_epsilon_0.1_P_p}, we compare the convergence property of the error probability of the proposed algorithm with respect to $\epsilon = 0.05$ and $d_0 = 200$ m for different levels of power of the pilot signals, $P = \{5, 10, 20\}$ dBm. As we increase the power of the pilot signals, the estimation accuracy of $\alpha_1$ and $\alpha_2$ increases and hence the error probability decreases. This is so because, as we increase the power of pilot signals the received signals will be less noisy which increases the chances of estimating the $\alpha_1$ and $\alpha_2$ more accurately.
\begin{figure}[ht]
\vspace{-5mm}
\centering
	\includegraphics[scale = 0.6]{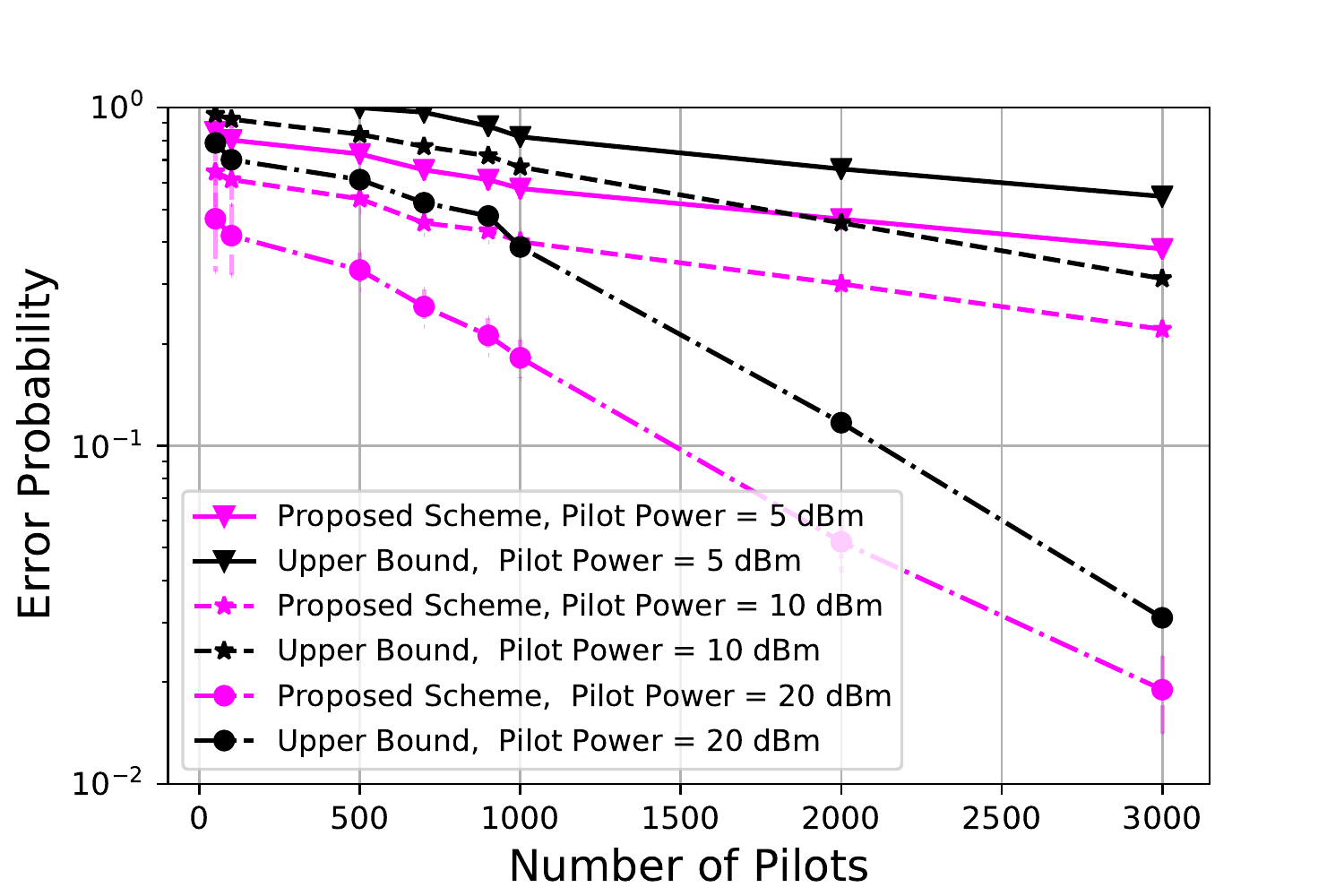}
\caption{\small{Error Probability Bound v/s Number of Pilots for $\epsilon = 0.05$ for $P = \{5, 10, 20\}$ dBm.}}
	\label{fig:Probability_Bound_alpha1_samples_epsilon_0.1_P_p}
	\vspace{-3mm}
\end{figure}

%% file: Conclusions.tex
We investigated the problem of estimation of the optimal phase-shift at the HMT-assisted wireless communication system in a noisy environment. We proposed a learning algorithm to estimate the optimal phase-shifting parameters and showed that the probability that the phase-shifting parameters generated by the proposed algorithm to deviate by more than $\epsilon$ from the optimal values decay exponentially fast as the number of pilots grows. Our proposed algorithm exploited structural properties of the channel gains in the far-field regions. 

%% file: Appendix.tex
\subsection{Proof of Proposition \ref{prop:prob_ineq}}
\label{subsec:Prop1}

\begin{proof}
Let us define the following events.
\begin{align}
&A_{n,m} = \abs{X_n - X_m}>\epsilon, \quad A_n = \abs{X_n - X}>\frac{\epsilon}{2}, \nonumber\\
&\quad \text{ and } \quad A_m = \abs{X_m - X}> \frac{\epsilon}{2}\nonumber
\end{align}
By the triangle inequality, we have
\begin{align}
\label{eqn:triangleineq}
   \abs{X_n - X_m} \leq \abs{X_n - X} + \abs{X_m - X}. 
\end{align}
Using \eqref{eqn:triangleineq}, the event $A_{n,m}$ can be written as
\begin{align}
&\abs{X_n - X_m} \geq \epsilon \implies \abs{X_n - X} + \abs{X_m - X} \geq \epsilon\nonumber
\end{align}
Therefore, we have
\begin{align}
 A_{n,m} &\subset \{|X_n - X|+|X - X_m|>\epsilon\} \nonumber\\
 &\subset \left\{|X_n - X|>\frac{\epsilon}{2} \bigcup |X - X_m|>\frac{\epsilon}{2}\right\}\label{eqn:set_bound}
\end{align}
Note that for any two events $A$ and $B$ where $A \subset B$, then $\Prob{A} \leq \Prob{B}$. We use this fact in \eqref{eqn:set_bound}, and we get
\begin{align}
& \Prob{|X_n - X_m|>\epsilon} \leq \Prob{|X_n - X|>\frac{\epsilon}{2}} + \Prob{|X_m - X|>\frac{\epsilon}{2}}\nonumber\qedhere
\end{align}
\end{proof}

\subsection{Proof of Lemma \ref{lem:chi_dist}}
\begin{proof}
We consider $r(\beta_1,\beta_2)$ as given in \eqref{eqn:absRSS} which comprises of two complex-valued factors $\sqrt{P}\times H(\beta_1,\beta_2)$ (see  \eqref{eqn:contphaseHMMIMOuserchannel}) and $\zeta$.

Write $\zeta = n_1 + jn_2$, where $n_1$ and $n_2$ follows $N\left(0, \frac{\sigma^2}{2}\right)$ and are independent, and write $\sqrt{P}\times H(\beta_1,\beta_2) = a + jb$, where $a$ and $b$ are real values. Therefore,
\begin{align}
\label{eqn:complexr}
    r(\beta_1,\beta_2) = \abs{y(\beta_1, \beta_2)}^2 = (a+n_1)^2 + (b+n_2)^2.
\end{align}
Note that $\frac{a+n_1}{\sigma/\sqrt{2}} \sim N\left(\frac{a}{\sigma/\sqrt{2}}, 1\right)$ and $\frac{b+n_2}{\sigma/\sqrt{2}} \sim N\left(\frac{b}{\sigma/\sqrt{2}}, 1\right)$ and they are independent. Therefore, 
\begin{align}
\label{eqn:sumofchi}
\frac{2}{\sigma^2}\bigg\{(a+n_1)^2 + (b+n_2)^2\bigg\} \sim \chi^2_{2}\left(\frac{2}{\sigma^2}\left(a^2+b^2\right)\right). 
\end{align}
Applying \eqref{eqn:sumofchi} in \eqref{eqn:complexr}, we get $X = \frac{2}{\sigma^2}r(\beta_1,\beta_2) \sim \chi^2_{2}\left(\lambda_1\right),$ where $\lambda_1 = \frac{2}{\sigma^2}\abs{\sqrt{P}\times H(\beta_1,\beta_2)}^2.$ The second part of the lemma follows from the additive property of non-central Chi-squared distribution of the sum of $n$ i.i.d. RVs of  $\chi^2_{2}\left(\lambda_1\right).$
\end{proof}

\subsection{Proof of Theorem 
\ref{thm:sub_exp_tail_bound}}
As $X_k, \forall k$ are independent, applying the definition \ref{defn:Subexponential}, the moment generating function of $\sum\limits_{k=1}^n(X_k - \mu_k)$ is given by
\begin{align}
\EE{e^{t\sum\limits_{k=1}^n(X_k - \mu_k)}} &\leq e^{\frac{\lambda^2}{2}\sum\limits_{k=1}^n \nu_k^2}, \quad \forall |t|<\left(\frac{1}{\max\limits_{k=1,2,\dots,n} b_k}\right).\nonumber
\end{align}
Since the moment generating functions uniquely determines the distribution, comparing with the definition \ref{defn:Subexponential}, it follows that
$\sum\limits_{k=1}^n(X_k - \mu_k)$ is a sub-exponential $(\nu_{*},b_{*})$ random variable, where
\[b_{*} = \max\limits_{k=1,2\dots,n} b_k \quad \text{and}\quad  \nu_{*} = \sqrt{\sum\limits_{k=1}^n \nu_k^2}.\]
To prove the second part of the Theorem we use the following tail bound on a sub-exponential distribution proved in \cite{wainwright2019high}.

\begin{prop}[\cite{wainwright2019high} Proposition 2.9]
\label{prop:sub_exp_tail_bound}
Let $X$ is sub-exponential random variable with parameters $(\nu,b)$ and $\EE{X}=\mu$. Then
\begin{align*}
   \Prob{\abs{X - \mu} \geq t} \leq
   \begin{cases}
    &2e^{-\frac{t^2}{2\nu^2}}, \qquad \text{if } 0 \leq t \leq \frac{\nu^2}{b}\\
    &2e^{-\frac{t}{2b}}, \qquad \text{if } t \geq \frac{\nu^2}{b}.
   \end{cases}
\end{align*}
\end{prop}
\noindent
The claim immediately follows by applying the above result on $Z_n:=\sum\limits_{k=1}^n(X_k - \mu_k),$ which is sub-exponential $(\nu_{*},b_{*}),$ where $b_{*} = \max\limits_{k=1,2\dots,n} b_k$  and $\nu_{*} = \sqrt{\sum\limits_{k=1}^n \nu_k^2}$. \hfill \IEEEQED

\subsection{Proof of Corollary \ref{cor:sub_exp_tail_bound}}
\label{subsec:ProofCor}
From Theorem~\ref{thm:sub_exp_tail_bound}, $\sum_{k=1}^n(X_k - \mu_k)$ is sub-exponential $(\nu_{*},b_{*}),$ where $b_{*} = 4$  and $\nu_{*} = 2\sqrt{n}(2+2a)$. Using the parameters $(\nu_*,b_*) = (2\sqrt{n}(2+2a),4)$ in Proposition \ref{prop:prob_ineq}, we get the required upper bound as
\begin{align}
&\Prob{\abs{\frac{1}{n}\sum\limits_{k=1}^n(X_k - \mu_k)} \geq t} \leq 
   \begin{cases}
    2e^{-\frac{nt^2}{8(2+2a)^2}}, &0\leq t \leq (2+2a)^2\nonumber\\
    2e^{-\frac{nt}{8}}, &t \geq (2+2a)^2\nonumber\\
    \end{cases}\\
 &\Prob{\abs{\frac{1}{n}\sum\limits_{k=1}^n(X_k - \mu_k)} \geq t} \leq     2e^{-\frac{n t^2}{8(2+2a)^2}}, \qquad t > 0. \nonumber
\end{align}

\hfill \IEEEQED

\subsection{Proof of Lemma \ref{lem:XSubExp}}
If $X\sim\chi_p^2(a)$, then according to \cite{el2013moment}, the moment-generating function (MGF) of $X$ is given by
\begin{align}
    \EE{\exp\{t(X-(p+a)\}} &= e^{-t(p+a)}\EE{e^{tX}}\nonumber\\
    &= \frac{e^{-t(p+a)}e^{\frac{at}{1-2t}}}{(1-2t)^{p/2}}\nonumber\\
    &= e^{\frac{2at^2}{1-2t}}\frac{e^{-pt}}{(1-2t)^{p/2}}, \qquad \text{ for } t < \frac{1}{2}.\label{eqn:MGF1}
\end{align}
By following some calculus, refer  \cite{ghosh2021exponential}, \cite[Example 2.8]{wainwright2019high}, we obtain
\begin{align}
    \frac{e^{-pt}}{(1-2t)^{p/2}} \leq e^{2pt^2}, \qquad \text{ for } \abs{t} \leq \frac{1}{4}. \label{eqn:MGF2}
\end{align}
For $\abs{t} \leq \frac{1}{4}$, we have
\begin{align}
    e^{\frac{2at^2}{1-2t}} \leq e^{4at^2}. \label{eqn:MGF3}
\end{align}
Applying \eqref{eqn:MGF2} and \eqref{eqn:MGF3} to \eqref{eqn:MGF1}, we obtain
\begin{align}
     \EE{\exp\{t(X-(p+a)\}} \leq e^{2(p+2a)t^2}, \qquad\forall \abs{t} \leq \frac{1}{4}.\label{eqn:MGF4}
\end{align}
Therefore, by \eqref{eqn:MGF4}, $X$ is Sub-exponential distribution with parameters $\big(2(p+2a), 4\big).$ 
\hfill \IEEEQED